\pdfoutput=1

\documentclass[12pt]{article}
\usepackage{amsmath, amssymb, amsfonts, amsthm, graphicx, calc, color}
\usepackage[multiple]{footmisc}
\usepackage{enumitem}
\usepackage{nag}
\renewcommand{\geq}{\geqslant}
\renewcommand{\ngeq}{\ngeqslant}
\renewcommand{\leq}{\leqslant}
\renewcommand{\nleq}{\nleqslant}

\usepackage[backend=biber,autolang=other,style=apa,maxcitenames=5,sortcites=false]{biblatex}
\DeclareLanguageMapping{british}{british-apa}
\DeclareLanguageMapping{american}{american-apa}
\usepackage[lmargin=1.5in,rmargin=1.5in,tmargin=1.25in,bmargin=1.25in]{geometry}

\usepackage{setspace}
\usepackage[colorlinks=true,linkcolor=blue,urlcolor=black,citecolor=black]{hyperref}
\hypersetup{pdfborder=0 0 0}

\usepackage[nameinlink]{cleveref}
\crefname{page}{p.}{pp.}
\crefname{equation}{equation}{equations}
\crefname{section}{section}{sections}
\crefname{subsection}{section}{sections}
\crefname{subsubsection}{section}{sections}
\crefname{appsec}{appendix}{appendices}
\crefname{supplsec}{supplemental appendix}{supplemental appendices}
\crefname{footnote}{footnote}{footnotes}
\crefname{figure}{figure}{figures}
\crefname{table}{table}{tables}
\crefname{theorem}{theorem}{theorems}
\crefname{proposition}{proposition}{propositions}
\crefname{lemma}{lemma}{lemmata}
\crefname{corollary}{corollary}{corollaries}
\crefname{remark}{remark}{remarks}
\crefname{observation}{observation}{observations}
\crefname{example}{example}{examples}
\crefname{fact}{fact}{facts}
\crefname{definition}{definition}{definitions}
\crefname{assumption}{assumption}{assumptions}
\crefname{exercise}{exercise}{exercises}
\crefname{notation}{notation}{notation}
\crefname{claim}{claim}{claims}
\crefname{conjecture}{conjecture}{conjectures}

\newtheorem{theorem}{\textbf{Theorem}}

\newtheorem{claim}{\textbf{Claim}}

\newtheorem{corollary}{\textbf{Corollary}}

\newtheorem{example}{\textbf{Example}}

\newtheorem{lemma}{\textbf{Lemma}}

\newtheorem{proposition}{\textbf{Proposition}}
\theoremstyle{definition}
\newtheorem{remark}{\textbf{Remark}}

\newtheoremstyle{named}
    {\topsep}                   % ABOVESPACE
    {\topsep}                   % BELOWSPACE
    {\itshape}                  % BODYFONT
    {0pt}                       % INDENT (empty value is the same as 0pt)
    {\bfseries}                 % HEADFONT
    {}                          % HEADPUNCT
    {5pt plus 1pt minus 1pt}    % HEADSPACE
    {\thmnote{#3}}              % CUSTOM-HEAD-SPEC
\theoremstyle{named}
\newtheorem{namedthm}{}

\renewcommand\qedsymbol{\bfseries QED}

\usepackage{fixfoot}
\DeclareFixedFootnote{\fnsuchan}{Such an $\bar x$ must exist, provided the argmax is nonempty (refer to \cref{footnote:LC_frictionless}).}
\DeclareFixedFootnote{\fnsuchmult}{Such $\underline x$ and $\bar x$ must exist, provided the argmaxes are nonempty (refer to \cref{footnote:LC_frictionless}).}

\newcommand{\R}{\mathbb R}
\newcommand{\N}{\mathbb N}
\newcommand{\Z}{\mathbb Z}
\newcommand{\E}{\mathbb E}
\DeclareMathOperator*{\argmax}{arg\,max}

\title{\vspace{-0.2cm}\protect\linespread{1}\protect\selectfont Comparative statics with adjustment costs\\ and the Le Chatelier principle%
\thanks{We are grateful for comments from Gregorio Curello, Christian Ewerhardt, Alkis Georgiadis-Harris, Alex Kohlhas, Bart Lipman, Hamish Low, Konstantin Milbradt, Paul Milgrom, Sara Neff, Marek Pycia, Simon Quinn, Ronny Razin, Karthik Sastry, Jakub Steiner, Quitzé Valenzuela-Stookey, Nathaniel Ver Steeg, Jan Žemlička, Mu Zhang, three anonymous referees, and audiences at Berlin, Bonn, Bristol, Caltech, CERGE-EI, Carlos III Madrid, Collegio Carlo Alberto, Columbia, Edinburgh, LSE, Manchester, National University of Singapore, Oxford, Paris--Saclay, Purdue, Queen Mary, Singapore Management University, Stony Brook, Tsinghua, UCLA, Western, Yale, Zürich, and several conferences. Dekel acknowledges financial support from the National Science Foundation and from the Foerder Institute at Tel Aviv University.}%
}

\author{Eddie Dekel\thanks{Department of Economics, Northwestern University and Department of Economics, Tel Aviv University. Email address: $<$\href{mailto:eddiedekel@gmail.com}{eddiedekel@gmail.com}$>$.}
\and John K.-H. Quah\thanks{Department of Economics, National University of Singapore. Email address: $<$\href{mailto:ecsqkhj@nus.edu.sg}{ecsqkhj@nus.edu.sg}$>$.}
\and Ludvig Sinander\thanks{Department of Economics and Nuffield College, University of Oxford. Email address: $<$\href{mailto:ludvig.sinander@economics.ox.ac.uk}{ludvig.sinander@economics.ox.ac.uk}$>$.}}

\date{20 October 2024}

\begin{document}

\maketitle

{\vspace{-0.6cm}\singlespacing
\begin{quote}
\noindent \textbf{Abstract:}\,
We develop a theory of monotone comparative statics for models with adjustment costs. We show that comparative-statics conclusions may be drawn under the usual ordinal complementarity assumptions on the objective function, assuming very little about costs: only a mild monotonicity condition is required. We use this insight to prove a general Le Chatelier principle: under the ordinal complementarity assumptions, if short-run adjustment is subject to a monotone cost, then the long-run response to a shock is greater than the short-run response. We extend these results to a fully dynamic model of adjustment over time: the Le Chatelier principle remains valid, and under slightly stronger assumptions, optimal adjustment follows a monotone path. We apply our results to models of saving, production, pricing, labor supply and investment.

\vspace{0.1cm}


\end{quote}
}

%%%%%%%%%%%%%%%%%%%%%%%%%%%%%%%%%%%
%%%%%%%%%%%%%%%%%%%%%%%%%%%%%%%%%%%
\section{Introduction}
\label{sec:intro}
%%%%%%%%%%%%%%%%%%%%%%%%%%%%%%%%%%%
%%%%%%%%%%%%%%%%%%%%%%%%%%%%%%%%%%%

Adjustment costs play a major role in explaining a wide range of economic phenomena. Examples include the investment behavior of firms,%
\footnote{E.g. \cite{Jorgenson1963,Hayashi1982}; Cooper \& Haltiwanger, \citeyear{CooperHaltiwanger2006}.}
price stickiness,%
\footnote{E.g. \cite{Mankiw1985}; Caplin \& Spulber, \citeyear{CaplinSpulber1987}; Golosov \& Lucas, \citeyear{GolosovLucas2007}; \cite{Midrigan2011}.}
trade in illiquid financial markets,%
\footnote{E.g. \cite{Kyle1985,Back1992}.}
aggregate consumption dynamics,%
\footnote{E.g. Kaplan \& Violante, \citeyear{KaplanViolante2014}; Berger \& Vavra, \citeyear{BergerVavra2015}; Chetty \& Szeidl, \citeyear{ChettySzeidl2016}.}
labor supply,%
\footnote{E.g. Chetty, Friedman, Olsen, \& Pistaferri, \citeyear{ChettyEtal2011}; \cite{Chetty2012}.}
labor demand,%
\footnote{E.g. \cite{Hamermesh1988}; Bentolila \& Bertola, \citeyear{BentolilaBertola1990}.}
and housing consumption and asset pricing.%
\footnote{Grossman \& Laroque, \citeyear{GrossmanLaroque1990}.}

In this paper, we develop a theory of monotone comparative statics with adjustment costs. Our fundamental insight is that very little needs to be assumed about the cost function: comparative statics requires only that \emph{not} adjusting be cheaper than adjusting, plus the usual ordinal complementarity assumptions on the objective function. We use this insight to show that Samuelson's (\citeyear{Samuelson1947}) \emph{Le Chatelier principle} is far more general than previously claimed: it holds whenever adjustment is costly, given only minimal structure on costs. We extend our comparative-statics and Le Chatelier results to a fully dynamic model of adjustment.

We apply our results to models of factor demand, capital investment, pricing, labor supply, and saving by wishful thinkers. These models are typically studied only under strong functional-form assumptions, and the cases of convex and nonconvex costs are considered separately and handled very differently. Our general results yield robust comparative statics for these standard models, dispensing with auxiliary assumptions and handling convex and nonconvex costs in a unified fashion.

The abstract setting is as follows. An agent chooses an action $x$ from a sublattice $L \subseteq \R^n$. Her objective $F(x,\theta)$ depends on a parameter $\theta$. At the initial parameter $\underline \theta$, the agent chose $\underline x \in \argmax_{x \in L} F(x,\underline \theta)$. The parameter now increases to $\bar \theta \geq \underline \theta$, and the agent may adjust her choice. Adjusting the action by $\varepsilon = x - \underline x$ costs $C(\varepsilon) \geq 0$, and the agent's new choice maximizes $G(x,\bar \theta) = F(x,\bar \theta) - C(x-\underline x)$.

Our only assumption on the cost function $C$ is \emph{monotonicity:}
\begin{equation*}
C(\varepsilon_1,\dots,\varepsilon_{i-1},\varepsilon_i',\varepsilon_{i+1},\dots,\varepsilon_n) \leq C(\varepsilon)
\quad \text{whenever $0 \leq \varepsilon'_i \leq \varepsilon_i$ or $0 \geq \varepsilon'_i \geq \varepsilon_i$.}
\end{equation*}
This means that cost falls whenever an adjustment vector $\varepsilon$ is modified by shifting one of its entries closer to zero (``no adjustment''). An additively separable cost function $C(\varepsilon) = \sum_{i=1}^n C_i(\varepsilon_i)$ is monotone if and only if each dimension's cost function $C_i$ is single-dipped and minimized at zero.

We allow some adjustments $\varepsilon$ to be infeasible, as captured by a prohibitive cost $C(\varepsilon)=\infty$.
In some of our results, the monotonicity assumption may be weakened to \emph{minimal monotonicity:} cost falls whenever an adjustment vector $\varepsilon$ is modified by replacing all of its positive entries with zero ($C( \varepsilon \wedge 0 ) \leq C(\varepsilon)$), and similarly for the negative entries ($C( \varepsilon \vee 0 ) \leq C(\varepsilon)$). In the additively separable case, this means that each dimension's cost $C_i$ is minimized at zero.

We eschew restrictive \emph{curvature} assumptions on costs, such as convexity; with multidimensional actions, even quasiconvexity is not needed. As we show, it is monotonicity-type properties, not convexity-type properties, which govern the direction of adjustment in response to a shock. (What curvature properties govern is the \emph{speed} of adjustment---an important but distinct question.)

Our basic question is under what assumptions on the objective $F$ and cost $C$ the agent's choice increases, in the sense that $\widehat x \geq \underline x$ for some $\widehat x \in \argmax_{x \in L} G(x,\bar \theta)$ (provided the argmax is not empty; such qualifiers are omitted throughout this introduction). Our fundamental result, \Cref{theorem:basic}, answers this question: nothing need be assumed about the cost $C$ except minimal monotonicity, while $F$ need only satisfy the ordinal complementarity conditions of \emph{quasi-supermodularity} and \emph{single-crossing differences} that feature in similar comparative-statics results absent adjustment costs \parencite[see][]{MilgromShannon1994}. Thus costs need not even be monotone, and the objective need not satisfy any cardinal properties, such as supermodularity or increasing differences. We also provide a generalization (\hyperref[theorem:basic-constraint]{\Cref*{theorem:basic}$^*$}) allowing for shifts of both the constraint set $L$ and of the parameter $\theta$, and we give a ``$\forall$'' variant (\Cref{proposition:basic-strict}) showing that adding either of two mild assumptions yields the stronger conclusion that $\widehat x \geq \underline x$ for \emph{every} $\widehat x \in \argmax_{x \in L} G(x,\bar \theta)$.

We use our fundamental result to re-think Samuelson's (\citeyear{Samuelson1947}) \emph{Le Chatelier principle,} which asserts that the response to a parameter shift is greater at longer horizons. Our \Cref{theorem:lechatelier} provides that the Le Chatelier principle holds whenever short-run adjustment is subject to a monotone adjustment cost $C$, long-run adjustment is frictionless, and the objective $F$ satisfies the ordinal complementarity conditions. Formally, the theorem states that under these assumptions, given any long-run choice $\bar x \in \argmax_{x \in L} F(x,\bar \theta)$ satisfying $\bar x \geq \underline x$, we have $\bar x \geq \widehat x \geq \underline x$ for some optimal short-run choice $\widehat x \in \argmax_{x \in L} G(x,\bar \theta)$.%
\footnote{Furthermore, if $\bar x$ is the largest element of $\argmax_{x \in L} F(x,\underline x)$, then $\bar x \geq \widehat x$ for \emph{any} short-run choice $\widehat x \in \argmax_{x \in L} G(x,\bar \theta)$.}
This substantially generalizes Milgrom and Roberts's (\citeyear{MilgromRoberts1996}) Le Chatelier principle, in which short-run adjustment is assumed to be impossible for some dimensions $i$ and costless for the rest: that is, $C(\varepsilon) = \sum_{i=1}^n C_i(\varepsilon_i)$, where some dimensions $i$ have $C_i(\varepsilon_i) = \infty$ for all $\varepsilon_i \neq 0$, and the rest have $C_i \equiv 0$. We show that our Le Chatelier principle remains valid if long-run adjustment is also costly (\Cref{proposition:lechatelier-medium}), and we identify two weak assumptions under either of which $\bar x \geq \widehat x \geq \underline x$ holds for \emph{every} $\widehat x \in \argmax_{x \in L} G(x,\bar \theta)$ (\Cref{proposition:lechatelier-strict}).

We then extend our comparative-statics and Le Chatelier theorems to a fully dynamic, forward-looking model of costly adjustment over time. The parameter $\theta_t$ evolves over time $t \in \{1,2,3,\dots\}$, and the adjustment cost function $C_t$ may also vary between periods. Starting at $x_0 = \underline x \in \argmax_{x \in L} F(x,\underline \theta)$, the agent chooses a path $(x_t)_{t=1}^\infty$ to maximize the discounted sum of her period payoffs $F(x_t,\theta_t) - C_t(x_t - x_{t-1})$. \Cref{theorem:lechatelier-dynamic} validates the Le Chatelier principle: under the same assumptions (ordinal complementarity of $F$ and monotonicity of each $C_t$), if $\underline \theta \leq \theta_t \leq \bar \theta$ in every period $t$, then given any $\bar x \in \argmax_{x \in L} F(x,\bar \theta)$ such that $\underline x \leq \bar x$, the agent's choices satisfy $\underline x \leq x_t \leq \bar x$ along some optimal path $(x_t)_{t=1}^\infty$. If the parameter and cost are time-invariant ($\theta_t = \bar \theta$ and $C_t = C$ for all periods $t$), then a stronger Le Chatelier principle holds (\Cref{theorem:lechatelier-dynamic-strong}): under additional assumptions, the agent adjusts more at longer horizons, in the sense that $\underline x \leq x_t \leq x_T \leq \bar x$ holds at any dates $t<T$ along some optimal path $(x_t)_{t=1}^\infty$.

The Le Chatelier principle remains valid if decisions are instead made by a sequence of short-lived agents (\Cref{theorem:lechatelier-dynamic-myopic}): under the same assumptions as in \Cref{theorem:lechatelier-dynamic}, if $\underline \theta \leq \theta_t \mathrel{(\leq \theta_{t+1}) \leq} \bar \theta$ in every period $t$, then $\underline x \leq \widetilde x_t \mathrel{(\leq \widetilde x_{t+1}) \leq} \bar x$ in every period $t$ along some equilibrium path $(\widetilde x_t)_{t=1}^\infty$. Thus short-lived agents adjust in the same direction as a long-lived agent would. They may do so more sluggishly, however: \Cref{theorem:myopic-vs-fwd} asserts that under stronger assumptions, $\underline x \leq \widetilde x_t \leq x_t \leq \bar x$ holds along some short-lived equilibrium path $(\widetilde x_t)_{t=1}^\infty$ and some long-lived optimal path $(x_t)_{t=1}^\infty$.

Several of our main results admit converses, which assert that monotonicity-type assumptions on adjustment costs are \emph{necessary} (as well as sufficient) for drawing comparative-statics conclusions. In particular, the cost assumptions in \Cref{theorem:basic,theorem:lechatelier-dynamic} are necessary as well as sufficient, while a condition slightly weaker than monotonicity is necessary and sufficient for \Cref{theorem:lechatelier}.

The rest of this paper is arranged as follows. In the next section, we describe the environment. We present our fundamental comparative-statics insight (\Cref{theorem:basic}) in \cref{sec:mcs}, and apply it to saving. In \cref{sec:lechatelier}, we develop a general Le Chatelier principle (\Cref{theorem:lechatelier}), and apply it to pricing and factor demand. In \cref{sec:dynamic}, we introduce a dynamic, forward-looking adjustment model, derive two dynamic Le Chatelier principles (\Cref{theorem:lechatelier-dynamic,theorem:lechatelier-dynamic-strong}), and apply them to pricing, labor supply, and investment. In \cref{sec:myopic}, we derive a Le Chatelier principle for short-lived agents (\Cref{theorem:lechatelier-dynamic-myopic}) and compare their behavior to that of a long-lived agent (\Cref{theorem:myopic-vs-fwd}). We conclude in \cref{sec:necessity} by establishing converses of several main results. The \hyperref[sec:appendix]{appendix} contains definitions of some standard terms, an extension to allow for uncertain adjustment costs, and all proofs omitted from the text.

%%%%%%%%%%%%%%%%%%%%%%%%%%%%%%%%%%%
%%%%%%%%%%%%%%%%%%%%%%%%%%%%%%%%%%%
\section{Setting}
\label{sec:setting}
%%%%%%%%%%%%%%%%%%%%%%%%%%%%%%%%%%%
%%%%%%%%%%%%%%%%%%%%%%%%%%%%%%%%%%%

The agent's objective is $F(x,\theta)$, where $x$ is the choice variable and $\theta \in \Theta$ is a parameter. The choice variable $x$ belongs to a subset $L$ of $\R^n$. (More generally, $x$ could be infinite-dimensional; our results apply also in that case. We focus on the finite-dimensional case merely for simplicity.)

At the initial parameter $\theta=\underline \theta$, an optimal choice $\underline x$ was made:
\begin{equation*}
\underline x\in\argmax_{x\in L} F(x,\underline \theta).
\end{equation*}
(Note that we allow for a multiplicity of optimal actions.) This is the agent's ``starting point,'' and we shall consider how she responds to a change in the parameter from $\underline \theta$ to $\bar \theta$, where $\bar \theta \geq \underline \theta$, when adjustment is costly.

Adjusting from $\underline{x}$ to $x$ costs $C(x-\underline{x})$. The cost function $C$ is a map $\Delta L\to [0,\infty]$, where $\Delta L=\{x-y\,:\,x, y\in L\}$. Note that we allow some adjustments $\varepsilon \in \Delta L$ to have infinite cost $C(\varepsilon)=\infty$, meaning that they are infeasible. We assume throughout that $C(0) < \infty$.

The agent adjusts her action $x \in L$ to maximize
\begin{equation*}
G(x,\bar \theta)
= F(x,\bar \theta) - C(x-\underline x).
\end{equation*}

\begin{remark}
Since $\underline x$ is held fixed, our assumption that cost depends only on $\varepsilon = x-\underline x$ is without loss of generality. In particular, if the ``true'' cost has the general form $\widetilde C(x,\underline x)$, then we interpret $C$ as $C(\varepsilon) = \widetilde C(\underline x + \varepsilon, \underline x )$.
\end{remark}

%%%%%%%%%%%%%%%%%%%%%%%%%%%%%%%%%%%
\subsection{Order assumptions}
\label{sec:setting:order_assns}
%%%%%%%%%%%%%%%%%%%%%%%%%%%%%%%%%%%

Throughout, $\R^n$ (and thus $L$) is endowed with the usual ``product'' order $\geq$, so ``$x \geq y$'' means ``$x_i \geq y_i$ for every dimension $i$.'' We write ``$x > y$'' whenever $x \geq y$ and $x \neq y$. We assume that the choice set $L$ is a \emph{sublattice} of $\R^n$, meaning that for any $x,y \in L$, the following two vectors also belong to $L$:
\begin{align*}
x \wedge y &= ( \min\{x_1,y_1\}, \dots, \min\{x_n,y_n\} )
\\
\text{and} \quad
x \vee y &= ( \max\{x_1,y_1\}, \dots, \max\{x_n,y_n\} ) .
\end{align*}
Examples of sublattices include $L = \R^n$, ``boxes'' $L = \{ x \in \R^n : y \leq x \leq z \}$ for $y,z \in \R^n$, ``grids'' such as $L = \Z^n$ (where $\Z$ denotes the integers), and half-planes $L = \{ (x_1,x_2) \in \R^2 : \alpha x_1 + \beta x_2 \geq k \}$ for $\alpha \leq 0 \leq \beta$ and $k \in \R$.

The parameter $\theta$ belongs to a partially ordered set $\Theta$. We use the symbol ``$\geq$'' also for the partial order on $\Theta$. In applications, the parameter $\theta$ is often a vector, in which case $\Theta$ is a subset of $\R^n$ and $\geq$ is the usual ``product'' order.

%%%%%%%%%%%%%%%%%%%%%%%%%%%%%%%%%%%
\subsection{Monotonicity assumptions on costs}
\label{sec:setting:cost_assns}
%%%%%%%%%%%%%%%%%%%%%%%%%%%%%%%%%%%

Most of our results assume that the cost $C$ is \emph{monotone,} but our first theorem requires only \emph{minimal monotonicity.} We now define these two properties.

The cost function $C$ is \emph{monotone} if and only if for any adjustment vector $\varepsilon \in \Delta L$ and any dimension $i$,
\begin{equation*}
C(\varepsilon_1,\dots,\varepsilon_{i-1},\varepsilon_i',\varepsilon_{i+1},\dots,\varepsilon_n) \leq C(\varepsilon)
\quad \text{whenever $0 \leq \varepsilon_i' \leq \varepsilon_i$ or $0 \geq \varepsilon_i' \geq \varepsilon_i$.}
\end{equation*}
In other words, modifying an adjustment vector by shifting one dimension's adjustment toward zero always reduces cost.
An equivalent definition of monotonicity is that $C(\varepsilon') \leq C(\varepsilon)$ holds whenever $\varepsilon'$ is ``between $0$ and $\varepsilon$'' in the sense that in each dimension $i$, we have either $0 \leq \varepsilon'_i \leq \varepsilon_i$ or $0 \geq \varepsilon'_i \geq \varepsilon_i$. Clearly monotonicity is an \emph{ordinal} property: if $C$ is monotone, then so is $\phi \circ C$ for any strictly increasing map $\phi : [0,\infty] \to [0,\infty]$.

If the choice variable is one-dimensional ($L \subseteq \R$), then monotonicity requires precisely that $C$ be single-dipped and minimized at zero.%
\footnote{Given $X \subseteq \R$, a function $\phi : X \to [0,\infty]$ is \emph{single-dipped} if and only if there is an $x \in X$ such that $\phi$ is decreasing on $\{ y \in X : y \leq x \}$ and increasing on $\{ y \in X : y \geq x \}$.}
More generally, if $C$ has the \emph{additively separable} form $C(\varepsilon) = \sum_{i=1}^n C_i(\varepsilon_i)$, then it is monotone if and only if each $C_i$ is single-dipped and minimized at zero.
%In general, monotonicity demands that \emph{holding fixed} the adjustments in all dimensions but $i$, the adjustment cost as a function of the dimension-$i$ adjustment $\varepsilon_i$ alone is single-dipped and minimized at zero.

\begin{example}
\label{example:costs_monotone_1d}
For a one-dimensional choice variable ($L \subseteq \R$), the following cost functions are monotone, for any values of the parameters $k,a \in (0,\infty)$:\,
(a)~Fixed cost: $C(\varepsilon) = k$ for $\varepsilon \neq 0$ and $C(0)=0$.\,
(b)~Quadratic cost: $C(\varepsilon) = a \varepsilon^2$.\,
(c)~Quadratic with free disposal: $C(\varepsilon) = a \varepsilon^2$ if $\varepsilon \geq 0$ and $C(\varepsilon) = 0$ otherwise.\,
(d)~Quadratic with a constraint: $C(\varepsilon) = a \varepsilon^2$ if $\varepsilon \in \mathcal{E}$ and $C(\varepsilon) = \infty$ otherwise, where the constraint set $\mathcal{E} \subseteq \R$ is convex and contains $0$.
\end{example}

\begin{example}
\label{example:costs_monotone_multid}
The following cost functions are monotone:\,
(a)~Additively separable: $C(\varepsilon) = \sum_{i=1}^n C_i(\varepsilon_i)$, where each $C_i$ is of one of the types in \Cref{example:costs_monotone_1d}.\,
(b)~Euclidean: $C(\varepsilon) = \sqrt{ \sum_{i=1}^n \varepsilon_i^2 }$.\,
(c)~Cobb--Douglas: $C(\varepsilon) = \prod_{i=1}^n {\lvert \varepsilon_i \rvert}^{a_i}$, where $a_1,\dots,a_n \in (0,\infty)$.
\end{example}

Monotonicity is consistent with quite general nonconvexities, and even with failures of quasiconvexity: the cost function in \Cref{example:costs_monotone_multid}(c) is monotone, but is not quasiconvex unless the choice variable is one-dimensional, i.e. $L \subseteq \R$. (Monotonicity \emph{does} imply quasiconvexity when the choice variable is one-dimensional, since then quasiconvexity is equivalent to single-dippedness.)

A cost function $C$ is called \emph{minimally monotone} if and only if
\begin{equation*}
C(\varepsilon \wedge 0) \leq C(\varepsilon) \geq C(\varepsilon\vee 0)
\quad \text{for any adjustment vector $\varepsilon \in \Delta L$.}
\end{equation*}
In other words, simultaneously \emph{cancelling} all upward adjustments, by replacing all of the positive entries of an adjustment vector $\varepsilon$ with zeroes, reduces cost; similarly, cancelling all downward adjustments reduces cost. Clearly minimal monotonicity is implied by monotonicity, and clearly it is an ordinal property.

If the choice variable is one-dimensional ($L \subseteq \R$), then minimal monotonicity demands exactly that $C$ be minimized at zero. If $C$ is additively separable, then it is minimally monotone if and only if each $C_i$ is minimized at zero.

\begin{example}
\label{example:costs_0monotone}
Consider the cost function in \Cref{example:costs_monotone_1d}(d), with a constraint set $\mathcal{E} \subseteq \R$ that that contains $0$ but is \emph{not} convex. For instance, $\mathcal{E} = (-\infty,0] \cup [I,\infty)$ for some $I>0$, as in the recent literature on the investment behavior of entrepreneurs in developing countries (see \cref{sec:dynamic:investment} below). Or $\mathcal{E} = \Z$ due to an integer constraint. Such a cost function is minimally monotone, but not monotone.
\end{example}

\begin{example}
If $C(x-\underline x) = 0$ whenever $x$ is a permutation of $\underline x$ and $C(x-\underline x) > 0$ otherwise, then $C$ is not minimally monotone (hence not monotone). This may occur if $x_1,\dots,x_n$ are prices at a firm's various establishments, swapping menus between establishments is costless, and printing new menus is costly.
\end{example}

%%%%%%%%%%%%%%%%%%%%%%%%%%%%%%%%%%%
\subsection{Complementarity assumptions on the objective}
\label{sec:setting:obj_assns}
%%%%%%%%%%%%%%%%%%%%%%%%%%%%%%%%%%%

We assume throughout that the objective function $F$ satisfies the standard \emph{ordinal complementarity conditions} of \emph{quasi-supermodularity} and \emph{single-crossing differences} \parencite[see][]{MilgromShannon1994}, defined as follows.

The objective $F(x,\theta)$ has \emph{single-crossing differences in $(x,\theta)$} if and only if $F(y,\theta')-F(x,\theta') \geq \mathrel{(>)} 0$ implies $F(y,\theta'')-F(x,\theta'') \geq \mathrel{(>)} 0$ whenever $x \leq y$ and $\theta' \leq \theta''$. Economically, this means that a higher parameter implies a greater liking for higher actions: whenever a higher action is (strictly) preferred to a lower one, this remains true if the parameter increases. A sufficient condition is \emph{increasing differences,} which requires that $F(y,\theta)-F(x,\theta)$ be increasing in $\theta$ whenever $x \leq y$. Related concepts, such as \emph{log} increasing differences, are defined in the \hyperref[sec:appendix]{appendix}.

A function $\phi : L \to \R$ is called \emph{quasi-supermodular} if $\phi(x) - \phi(x \wedge y) \geq \mathrel{(>)} 0$ implies $\phi(x \vee y) - \phi(y) \geq \mathrel{(>)} 0$. A sufficient condition is \emph{supermodularity,} which requires that $\phi(x) - \phi(x \wedge y) \leq \phi(x \vee y) - \phi(y)$ for any $x,y \in L$. If $L \subseteq \R$, then every function $\phi : L \to \R$ is automatically supermodular. See the \hyperref[sec:appendix]{appendix} for discussion and for definitions of related concepts, such as \emph{submodularity.}

We say that $F(x,\theta)$ is \emph{(quasi-)supermodular in $x$} if for each parameter $\theta \in \Theta$, the function $F(\cdot,\theta) : L \to \R$ is (quasi-)supermodular. This captures complementarity between the different dimensions of the action.

\begin{remark}\label{remark:karlin}
In applications to decision under uncertainty, the objective is typically $F(x,\theta) = \int_{\mathcal{S}} f(x,s,\theta) H(\mathrm{d}s,\theta)$, where $f(x,s,\theta)$ is the payoff of action $x \in L$ contingent on state $s \in \mathcal{S} \subseteq \R^k$, and $H(\cdot,\theta)$ is the CDF from which the state is drawn. Clearly $F(x,\theta)$ is supermodular in $x$ if $f(x,s,\theta)$ is. If $\mathcal{S} \subseteq \R$, $f(x,s,\theta)$ is independent of $\theta$ and has single-crossing differences in $(x,s)$, and $H(\cdot,\theta)$ increases with $\theta$ in the monotone likelihood ratio order, then $F(x,\theta)$ has single-crossing differences in $(x,\theta)$ \parencite[Lemma~1]{KarlinRubin1956}. For more conditions under which $F(x,\theta)$ is quasi-supermodular in $x$ or has single-crossing differences in $(x,\theta)$, see \textcite{Athey2002,QuahStrulovici2012}.
\end{remark}

%%%%%%%%%%%%%%%%%%%%%%%%%%%%%%%%%%%
%%%%%%%%%%%%%%%%%%%%%%%%%%%%%%%%%%%
\section{Comparative statics}
\label{sec:mcs}
%%%%%%%%%%%%%%%%%%%%%%%%%%%%%%%%%%%
%%%%%%%%%%%%%%%%%%%%%%%%%%%%%%%%%%%

Recall that the agent chooses $x \in L$ to maximize $G(x,\bar \theta)=F(x,\bar \theta)-C(x-\underline x)$. Our fundamental comparative-statics result is the following.

\begin{theorem} \label{theorem:basic}
Suppose that the objective $F(x,\theta)$ is quasi-supermodular in $x$ and has single-crossing differences in $(x,\theta)$, and that the adjustment cost $C$ is minimally monotone. If $\bar \theta \geq \underline \theta$, then $\widehat x\geq \underline x$ for some $\widehat x\in\argmax_{x\in L} G(x,\bar \theta)$, provided the argmax is nonempty.
\end{theorem}

In words, an increased parameter leads to a higher action (modulo tie-breaking). This parallels the basic comparative-statics result for costless adjustment \parencite[see][Theorem~4]{MilgromShannon1994}, one version of which states that under the same ordinal complementarity conditions on the objective $F$, we have $\bar x\geq \underline x$ for some $\bar x\in\argmax_{x\in L} F(x,\bar \theta)$, provided the argmax is nonempty. \Cref{theorem:basic} shows that this basic result is strikingly robust to adjustment costs: the objective $F$ need not satisfy any additional property, and the cost $C$ need only be minimally monotone.

\begin{example}
In many economic models, e.g. neoclassical production (see \cref{sec:lechatelier:factor_demand,sec:dynamic:investment} below), $F(x,\theta) = \phi(x) + \theta \cdot x$, where $L$ and $\Theta$ are subsets of $\R^n$ and $\phi : L \to \R$ is supermodular. Then $F(x,\theta)$ is supermodular in $x$ and has increasing differences in $(x,\theta)$. Hence, provided the adjustment cost is minimally monotone, \Cref{theorem:basic} guarantees that any coordinatewise increase of the parameter~$\theta$ leads to an coordinatewise higher optimal choice of action~$x$.
\end{example}

\Cref{theorem:basic} does \emph{not} follow from applying the basic comparative-statics result to the objective function $G(x,\bar \theta)$, because its assumptions do not guarantee that $G(\cdot,\bar \theta)$ is quasi-supermodular.%
\footnote{Its second term $x \mapsto -C(x-\underline x)$ need not be quasi-supermodular, and in any case, the sum of two quasi-supermodular functions is not quasi-supermodular in general.}
A different argument is required.

\begin{proof}[\normalfont\bfseries Proof]
Let $x'\in \argmax_{x\in L} G(x,\bar \theta)$. We claim that $\widehat x = \underline x\vee x'$ also maximizes $G(\cdot,\bar \theta)$; obviously $\widehat x \geq \underline x$. We have $F(\underline x,\underline \theta) \geq F(\underline x \wedge x', \underline \theta)$ by definition of $\underline x$. Thus $F(\underline x \vee x',\underline \theta) \geq F(x', \underline \theta)$ by quasi-supermodularity, whence $F(\underline x \vee x',\bar \theta) \geq F(x', \bar \theta)$ by single-crossing differences. Furthermore, by minimal monotonicity, $C( \underline x \vee x' - \underline x ) = C( ( x' - \underline x ) \vee 0 ) \leq C( x' - \underline x )$.
Thus
\begin{equation*}
G(\widehat x,\bar \theta)
= F(\underline x\vee x',\bar \theta)-C(\underline x\vee x'-\underline x)
\geq F(x',\bar \theta)-C(x'-\underline x)
= G(x',\bar \theta).
\end{equation*}
Since $x'$ maximizes $G(\cdot,\bar \theta)$ on $L$, it follows that $\widehat x$ does, too.
\end{proof}

This proof illustrates the role played by the minimal-monotonicity assumption in delivering comparative statics. This assumption cannot be weakened: we show in \cref{sec:necessity} below that minimal monotonicity is necessary (as well as sufficient) for comparative statics to hold whatever the objective $F$.

In applications, it is often useful that \Cref{theorem:basic} requires $F$ to satisfy only the ordinal complementarity conditions, rather than the stronger \emph{cardinal complementarity conditions} of supermodularity and increasing differences. In monopoly pricing, for example, the objective $F$ has single-crossing differences, but not increasing differences---see \cref{sec:lechatelier:pricing} below.

\Cref{theorem:basic} has a counterpart for parameter \emph{decreases:} under the same assumptions, if $\bar \theta \leq \underline \theta$, then $\widehat x\leq \underline x$ for some $\widehat x\in\argmax_{x\in L} G(x,\bar \theta)$, provided the argmax is nonempty. The proof is exactly analogous.%
\footnote{\label{footnote:one-half}The proof of \Cref{theorem:basic} uses only one half of the minimal-monotonicity assumption: that $C( \varepsilon \vee 0 ) \leq C(\varepsilon)$ for every $\varepsilon \in \Delta L$. The proof of its parameter-decrease counterpart uses (only) the other half, namely that $C( \varepsilon \wedge 0 ) \leq C(\varepsilon)$ for every $\varepsilon \in \Delta L$.}
All of the results in this paper have such counterparts for parameter decreases; we will not discuss them explicitly.

\Cref{theorem:basic} also extends straightforwardly to the case in which adjustment costs are uncertain, even if the agent is risk-averse. To be precise, let all uncertainty be summarized by a random variable $S$, called ``the state of the world.'' The agent's adjustment cost is $C_s(\cdot)$ in state $S=s$. Her ex-ante payoff is
\begin{equation*}
\widetilde G(x,\theta)
= \E\bigl[ u\bigl(
F(x,\theta) - C_S(x-\underline x)
\bigr) \bigr] ,
\end{equation*}
where $u$ is an increasing function $\R \to \R$, whose curvature captures the agent's risk attitude. \Cref{theorem:basic} remains true verbatim, except with ``$C$ is minimally monotone'' replaced by ``$C_s$ is minimally monotone for almost every realization $s$ of the state $S$,'' and $G$ replaced by $\widetilde G$. Several subsequent results also generalize along these lines. These assertions are proved in \cref{sec:appendix:uncertain}.

%%%%%%%%%%%%%%%%%%%%%%%%%%%%%%%%%%%
\subsection{Shifts of the constraint set}
\label{sec:mcs:constraint}
%%%%%%%%%%%%%%%%%%%%%%%%%%%%%%%%%%%

We now generalize \Cref{theorem:basic} to encompass shifts of the constraint set $L$ as well as of the parameter $\theta$. Write $\underline L \subseteq \R^n$ for the initial constraint set, $\underline x \in \argmax_{x \in \underline L} F(x,\underline \theta)$ for the agent's initial choice, and $\bar L \subseteq \R^n$ for the new constraint set. Recall that for two sets $X,Y \subseteq \R^n$, $X$ is \emph{higher in the strong set order} than $Y$, denoted $X \geq_{\text{ss}} Y$, if and only if for any $x \in X$ and $y \in Y$, the vector $x \vee y$ belongs to $X$ and the vector $x \wedge y$ belongs to $Y$.

\begin{namedthm}[\Cref*{theorem:basic}$\boldsymbol{^*}$.]\label{theorem:basic-constraint}
Suppose that the objective $F(x,\theta)$ is quasi-supermodular in $x$ and has single-crossing differences in $(x,\theta)$, and that the adjustment cost $C$ is minimally monotone. If $\bar \theta \geq \underline \theta$ and $\bar L \geq_{\text{ss}} \underline L$, then $\widehat x\geq \underline x$ for some $\widehat x\in\argmax_{x\in \bar L} G(x,\bar \theta)$, provided the argmax is nonempty.
\end{namedthm}

\begin{remark}\label{remark:sso}
\hyperref[theorem:basic]{Theorems~\ref*{theorem:basic}} and \hyperref[theorem:basic-constraint]{\ref*{theorem:basic}$^*$} are phrased differently than the usual statement of the basic result \parencite[see][Theorem~4]{MilgromShannon1994}, which asserts that when $F(x,\theta)$ is quasi-supermodular in $x$ and has single-crossing differences in $(x,\theta)$, if $\bar \theta \geq \underline \theta$ and $\bar L \geq_{\text{ss}} \underline L$ then $\argmax_{x \in \bar L} F(x,\bar \theta) \geq_{\text{ss}} \argmax_{x \in \underline L} F(x,\underline \theta)$. A version of \hyperref[theorem:basic-constraint]{\Cref*{theorem:basic}$^*$} with this form also holds: under the same hypotheses, if $\bar \theta \geq \underline \theta$ and $\bar L \geq_{\text{ss}} \underline L$ then 
\begin{align*}
\underline x
\in \argmax_{x \in \underline L} F(x,\underline \theta)
\qquad &\text{and} \:\;
&x'
\in \argmax_{x \in \bar L} G(x,\bar \theta) \phantom{.}
\\
\implies \quad\:\;
\underline x \wedge x'
\in \argmax_{x \in \underline L} F(x,\underline \theta)
\qquad &\text{and} \:\;
& \underline x \vee x'
\in \argmax_{x \in \bar L} G(x,\bar \theta) .
\end{align*}
The latter claim (about $\underline x \vee x'$) is exactly what the proof of \hyperref[theorem:basic-constraint]{\Cref*{theorem:basic}$^*$} shows. To see why $\underline x \wedge x' \in \argmax_{x \in \underline L} F(x,\underline \theta)$, suppose not; then $F(\underline x,\underline \theta) > F(\underline x \wedge x', \underline \theta)$, so that replicating the steps in the proof of \hyperref[theorem:basic-constraint]{\Cref*{theorem:basic}$^*$} delivers $G(\underline x \vee x',\bar \theta) > G(x',\bar \theta)$, which contradicts the fact that $x'$ maximizes $G(\cdot,\bar \theta)$ on $\bar L$.
\end{remark}

%%%%%%%%%%%%%%%%%%%%%%%%%%%%%%%%%%%
\subsection["For all" comparative statics]{``$\boldsymbol{\forall}$'' comparative statics}
\label{sec:mcs:strict}
%%%%%%%%%%%%%%%%%%%%%%%%%%%%%%%%%%%

We now provide a ``$\forall$'' counterpart to \Cref{theorem:basic}, giving two conditions under either of which $\widehat x \geq \underline x$ holds for \emph{every} optimal choice $\widehat x$. The first of these conditions is \emph{strict single-crossing differences} of the objective $F(x,\theta)$ in $(x,\theta)$, which requires that $F(y,\theta')-F(x,\theta') \geq 0$ implies $F(y,\theta'')-F(x,\theta'') > 0$ whenever $x < y$ and $\theta' < \theta''$. The second is \emph{strict minimal monotonicity} of the cost $C$, which demands that for any adjustment vector $\varepsilon \in \Delta L$,
\begin{equation*}
C(\varepsilon \wedge 0) < C(\varepsilon)
\quad \text{unless $\varepsilon \leq 0$,}
\quad \text{and} \quad
C(\varepsilon \vee 0) < C(\varepsilon)
\quad \text{unless $\varepsilon \geq 0$.\footnotemark}
\end{equation*}
\footnotetext{Equivalently: $C(\varepsilon \wedge 0) < C(\varepsilon)$ unless $\varepsilon \wedge 0 = \varepsilon$, and $C(\varepsilon \vee 0) < C(\varepsilon)$ unless $\varepsilon \vee 0 = \varepsilon$.}%
In other words, simultaneously cancelling all upward adjustments, by replacing all of the strictly positive entries of an adjustment vector $\varepsilon$ with zeroes, \emph{strictly} reduces cost; and likewise for downward adjustments.

\begin{proposition} \label{proposition:basic-strict}
Suppose that the objective $F(x,\theta)$ is quasi-supermodular in $x$, and that either
\begin{enumerate}[label=(\alph*)]
\item \label{item:strict-scd}
the objective $F(x,\theta)$ has \underline{strict} single-crossing differences in $(x,\theta)$ and the cost $C$ is minimally monotone, or
\item \label{item:strict-reduction}
the objective $F(x,\theta)$ has single-crossing differences in $(x,\theta)$ and the cost $C$ is \underline{strictly} minimally monotone.
\end{enumerate}
If $\bar \theta > \underline \theta$, then $\widehat x\geq \underline x$ for \underline{any} $\widehat x\in\argmax_{x\in L} G(x,\bar \theta)$.%
\footnote{A variant of \Cref{proposition:basic-strict} can be obtained by mixing the ``strictness'' properties~\ref{item:strict-scd} and \ref{item:strict-reduction}: if $x=(y,z)$, where $F(y,z,\theta)$ has strict single-crossing differences in $(y,\theta)$ for any fixed $z$, and $C(y-\underline{y},\cdot)$ is strictly minimally monotone for any fixed $y$, then the conclusion goes through, with essentially the same proof.}
\end{proposition}

\Cref{proposition:basic-strict} is the costly-adjustment analog of the standard ``$\forall$'' com\-pa\-ra\-tive-statics result \parencite[see][Theorem~$4^\prime$]{MilgromShannon1994}, which states that given any $\underline x \in \argmax_{x \in L} F(x,\underline \theta)$, if $F(x,\theta)$ is quasi-supermodular in $x$ and has strict single-crossing differences in $(x,\theta)$, then $\bar x \geq \underline x$ for \emph{any} $\bar x \in \argmax_{x \in L} F(x,\bar \theta)$. Part~\ref{item:strict-scd} directly extends this result to the costly-adjustment case. Part~\ref{item:strict-reduction} shows that the ``strictness'' in the hypotheses required to obtain a ``$\forall$'' comparative-statics conclusion can come from the cost $C$ rather than the objective $F$: in particular, strict minimal monotonicity ensures that even if some action $x \ngeq \underline x$ maximizes $F(\cdot,\bar \theta)$, it will not be chosen on account of its cost.

%%%%%%%%%%%%%%%%%%%%%%%%%%%%%%%%%%%
\subsection{Application to saving by wishful thinkers}
\label{sec:mcs:wishful}
%%%%%%%%%%%%%%%%%%%%%%%%%%%%%%%%%%%

An extensive literature in psychology and economics documents the prevalence of \emph{motivated reasoning:} believing (to some extent) what it is convenient to believe. Theoretical work has modeled this as an agent (perhaps subconsciously) \emph{choosing} her belief, balancing convenience against the costs of inaccuracy.%
    \footnote{This literature is surveyed by \textcite[][]{BenabouTirole2016,Benabou2015}.}

\emph{Wishful thinking} is motivated reasoning driven by a desire to be optimistic. \textcite{CaplinLeahy2019} study the economic implications of wishful thinking, showing (among other things) that wishful thinking may suppress saving in a standard consumption--saving model. In this section, we show how \Cref{theorem:basic} delivers this result without the authors' functional-forms assumptions.

An agent starts period 1 with wealth $w>0$. She consumes $c \in [0,w]$ and saves the rest. Savings accrue interest at rate $r>0$. The agent's period-2 income is uncertain, drawn from a finite set $\mathcal{Y} \subseteq \R_+$. Given her belief $G$ (a CDF on $\mathcal{Y}$) about her period-2 income, the agent's expected lifetime payoff is
\begin{equation*}
U(c,G) = u_1(c) + \int_{\mathcal{Y}} u_2\bigl((1+r)(w-c)+y\bigr) G(\mathrm{d} y) ,
\end{equation*}
where $u_1$ and $u_2$ are continuous, concave and strictly increasing. (A natural special case is when $u_2 = \delta u_1$ for some $\delta \in (0,1)$.)

In deciding what to believe about period-2 income, the agent contemplates a set $\mathcal{G}$ of beliefs (CDFs on $\mathcal{Y}$). We compare beliefs according to optimism, formalized by first-order stochastic dominance: $G \leq_1 H$ if and only if $G(y) \geq H(y)$ for every $y \in \mathcal{Y}$. We assume that $\mathcal{G}$ is a sublattice: if $G,H \in \mathcal{G}$, then $G \wedge_1 H$ and $G \vee_1 H$ also belong to $\mathcal{G}$, where $(G \wedge_1 H)(y) = \max\{G(y),H(y)\}$ and $(G \vee_1 H)(y) = \min\{G(y),H(y)\}$ for each $y \in \mathcal{Y}$. We further assume that $\mathcal{G}$ has a most optimistic element $\bar G$ (namely, $\bar G(y) = \inf_{G \in \mathcal{G}} G(y)$ for each $y \in \mathcal{Y}$).

A \emph{realist} holds belief $G_0 \in \mathcal{G}$, so consumes $c_0 \in \argmax_{c \in [0,w]} U(c,G_0)$. A \emph{wishful thinker} chooses both what to believe and how much to consume:
\begin{equation*}
\left( \widehat c, \widehat G \right)
\in \argmax_{(c,G) \in [0,w] \times \mathcal{G}}
\left[ U(c,G) - C(G-G_0) \right] ,
\end{equation*}
where $C : \Delta \mathcal{G} \to [0,\infty]$ is a minimally monotone cost function. An example is when $C(G-G_0)$ is the Kullback--Leibler divergence of $G$ from $G_0$;%
\footnote{\label{footnote:kl_min_mon}That is, $C(\varepsilon) = D(G_0+\varepsilon,G_0)$ for every $\varepsilon \in \Delta \mathcal{G}$, where $D$ is the Kullback--Leibler divergence. $C$ is minimally monotone since for any $G$, $D(G,G_0) = D(G,G_0) + D(G_0,G_0) \geq D(G_0 \wedge_1 G,G_0) + D(G_0 \vee_1 G,G_0)$ since $D(\cdot,G_0)$ is submodular \parencite[see][Example 16]{DziewulskiQuah2024}, which since $D \geq 0$ implies $D(G,G_0) \geq \max\{ D(G_0 \wedge_1 G,G_0), D(G_0 \vee_1 G,G_0) \}$.}
this is the functional form assumed by \textcite{CaplinLeahy2019}.

We claim that wishful thinkers save less than realists: $\bar c \geq \widehat c \geq c_0$ and $\bar G \geq_1 \widehat G \geq_1 G_0$, where $\bar c \in \argmax_{c \in [0,w]} U\left( c, \bar G \right)$ is how much the agent would consume if she were to hold the most optimistic belief $\bar G$. Formally:

\begin{proposition}
\label{proposition:wishful}
In the wishful-thinking application, under the stated assumptions, $\bar c \geq \widehat c \geq c_0$ and $\bar G \geq_1 \widehat G \geq_1 G_0$ for some
\begin{equation*}
\left( \widehat c, \widehat G \right) \in \argmax_{(c,G) \in [0,w] \times \mathcal{G}}
\left[ U(c,G) - C(G-G_0) \right]
\quad \text{and} \quad
\bar c \in \argmax_{c \in [0,w]} U\left( c, \bar G \right) ,
\end{equation*}
provided the argmaxes are nonempty.
\end{proposition}

The proof hinges on \hyperref[theorem:basic-constraint]{\Cref*{theorem:basic}$^*$}. It is important here that \hyperref[theorem:basic-constraint]{\Cref*{theorem:basic}$^*$} demands only \emph{minimal} monotonicity, because Caplin and Leahy's (\citeyear{CaplinLeahy2019}) Kull\-back--Leibler functional form does \emph{not} satisfy full-blown monotonicity.%
\footnote{Let $\mathcal{Y} = \{1,2,3\}$, and consider beliefs $G_0,G,H$ given by $(G_0(1),G_0(2),G_0(3)) = (\frac{1}{3},\frac{2}{3},1)$, $(G(1),G(2),G(3)) = (\frac{1}{4},\frac{1}{4},1)$ and $(H(1),H(2),H(3)) = (\frac{1}{8},\frac{1}{4},1)$. Then $G_0 \leq_1 G \leq_1 H$, but $C(G-G_0) - C(H-G_0) = \frac{1}{4} \ln(\frac{3}{4}) - 2 \times \frac{1}{8} \ln(\frac{3}{8}) = \frac{1}{4} \ln 2 >0$.}

\begin{proof}[\normalfont\bfseries Proof] Observe that $[0,w] \times \mathcal{G}$ is a sublattice.

\begin{claim}
\label{claim:wishful_spm}
$U$ is supermodular: $U(c,G) - U(c \wedge c', G \wedge_1 H) \leq U(c \vee c', G \vee_1 H) - U(c',H)$ for any $c,c' \in [0,w]$ and any beliefs $G,H \in \mathcal{G}$.
\end{claim}

The proof of \Cref{claim:wishful_spm} (\cref{sec:appendix:pf_wishful_spm}) turns on the concavity of $u_2$.

Assume that the argmaxes are nonempty. By \Cref{claim:wishful_spm}, $U(c,G)$ has increasing differences in $(c,G)$, i.e. $c \mapsto U(c,H) - U(c,G)$ is increasing if $G \leq_1 H$. Hence by the basic comparative-statics result \parencite[see][Theorem~4]{MilgromShannon1994}, we may choose a $\bar c \in \argmax_{c \in [0,w]} U\left( c, \bar G \right)$ such that $\bar c \geq c_0$. By the basic result again, if $\bar G \geq_1 \widehat G \geq_1 G_0$ then $\bar c \geq \widehat c \geq c_0$ for some
\begin{equation*}
\widehat c 
\in \argmax_{c \in [0,w]} U\left(c,\widehat G\right)
= \argmax_{c \in [0,w]}
\left[ U\left(c,\widehat G\right) - C\left(\widehat G-G_0\right) \right] .
\end{equation*}
It remains only to show that $\bar G \geq_1 \widehat G \geq_1 G_0$ for some
\begin{equation*}
\widehat G \in \argmax_{G \in \mathcal{G}} \left[ F(G) - C(G-G_0) \right] ,
\end{equation*}
where $F : \mathcal{G} \to \R$ is given by $F(G) = \max_{c \in [0,w]} U(c,G)$ for each $G \in \mathcal{G}$.

Note that $F$ is increasing: $F(G) \leq F(H)$ whenever $G \leq_1 H$. Hence $G_0 \in \argmax_{G \in \mathcal{G}_0} F(G)$, where $\mathcal{G}_0 = \left\{ G \in \mathcal{G} : G \leq_1 G_0 \right\}$. \Cref{claim:wishful_spm} implies that $F$ is supermodular \parencite[see][Theorem~2.7.6]{Topkis1998}. $C$ is minimally monotone, and clearly $\mathcal{G} \geq_{\text{ss}} \mathcal{G}_0$. Hence by \hyperref[theorem:basic-constraint]{\Cref*{theorem:basic}$^*$}, there is a 
\begin{equation*}
\widehat G \in \argmax_{G \in \mathcal{G}} \left[ F(G) - C(G-G_0) \right] 
\end{equation*}
such that $\widehat G \geq_1 G_0$. We have $\bar G \geq_1 \widehat G$ by definition of $\bar G$.
\end{proof}

%%%%%%%%%%%%%%%%%%%%%%%%%%%%%%%%%%%
%%%%%%%%%%%%%%%%%%%%%%%%%%%%%%%%%%%
\section{The Le Chatelier principle}
\label{sec:lechatelier}
%%%%%%%%%%%%%%%%%%%%%%%%%%%%%%%%%%%
%%%%%%%%%%%%%%%%%%%%%%%%%%%%%%%%%%%

The Le Chatelier principle asserts that an agent will adjust less (in every dimension) if subjected to an adjustment friction. A common interpretation equates frictional adjustment with the ``short run'' and frictionless adjustment with the ``long run,'' making the principle a claim about how the action response to parameter changes varies with the horizon.

In this section, we show that the Le Chatelier principle is far more general than previously claimed: it arises whenever the friction takes the form of a monotone adjustment cost. The classic formalization, which models friction as a constraint whereby some dimensions cannot be adjusted at all, is the special case in which each dimension has an adjustment cost that is either prohibitively high or equal to zero.

We shall compare the agent's response to a shift of the parameter from $\underline \theta$ to $\bar \theta$ in two cases: the case in which adjustment is costly, so that the agent's choice $\widehat x$ maximizes $G(\cdot,\bar \theta)$, and the case in which adjustment is costless, so the agent chooses a frictionless optimum $\bar x \in \argmax_{x \in L} F(x,\bar \theta)$. We call these cases ``short run'' and ``long run,'' respectively.%
    \footnote{Even in settings with richer dynamics, comparing adjustment between these two cases can be useful: we give an infinite-horizon example in \Cref{remark:dynamic_initial_static} below (\cref{sec:dynamic:setting}).}

Recall from \cref{sec:setting:cost_assns} the definition of a \emph{monotone} cost function $C$.

\begin{theorem}[Le Chatelier principle] \label{theorem:lechatelier}
Suppose that the objective $F(x,\theta)$ is quasi-supermodular in $x$ and has single-crossing differences in $(x,\theta)$, and that the adjustment cost $C$ is monotone. Fix $\bar \theta \geq \underline \theta$, and let $\bar x\in \argmax_{x\in L} F(x,\bar \theta)$ satisfy $\bar x\geq \underline x$.%
\footnote{\label{footnote:LC_frictionless}Such an $\bar x$ must exist, provided the argmax is nonempty. This follows from the basic comparative-statics result \parencite[see][Theorem~4]{MilgromShannon1994}.}
Then
\begin{itemize}
\item $\bar x \geq \widehat x \geq \underline x$ for some $\widehat x\in \argmax_{x\in L} G(x,\bar \theta)$, provided the argmax is nonempty, and
\item if $\bar x$ is the largest element of $\argmax_{x\in L} F(x,\bar \theta)$, then $\bar x \geq \widehat x$ for any $\widehat x\in \argmax_{x\in L} G(x,\bar \theta)$.
\end{itemize}
\end{theorem}

\Cref{theorem:lechatelier} nests the Le Chatelier principle of \textcite{MilgromRoberts1996}, in which it is assumed that only some dimensions $x_i$ of the choice variable can be adjusted in the short run, and that such adjustments are costless. This is the special case of our model in which $C(\varepsilon) = \sum_{i=1}^n C_i(\varepsilon_i)$, where some dimensions $i$ have $C_i \equiv 0$, and the other dimensions $i$ have $C_i(\varepsilon_i) = \infty$ for every $\varepsilon_i \neq 0$.

Like \Cref{theorem:basic}, \Cref{theorem:lechatelier} requires $F$ only to satisfy ordinal complementarity properties, not cardinal ones. This greatly extends its applicability, allowing it to be used to study pricing, for example (see \cref{sec:lechatelier:pricing} below).

\begin{proof}[\normalfont\bfseries Proof]
For the first part, assume that $\argmax_{x\in L} G(x,\bar \theta)$ is nonempty. By \Cref{theorem:basic}, we may choose an $x' \in \argmax_{x\in L} G(x,\bar \theta)$ such that $x'\geq \underline x$. We claim that $\widehat x = \bar x\wedge x'$ also maximizes $G(\cdot,\bar \theta)$; this suffices since $\bar x \geq \widehat x \geq \underline x$. We have $F(\bar x\vee x',\bar \theta) \leq F(\bar x,\bar \theta)$ by definition of $\bar x$, which by quasi-supermodularity implies that $F(x',\bar \theta) \leq F(\bar x\wedge x',\bar \theta)$. Since $C$ is monotone and $x' \geq \bar x \wedge x' \geq \underline x$, we have $C(x'-\underline x)\geq C(\bar x\wedge x'-\underline x)$. Thus
\begin{equation*}
G(x',\bar \theta)
= F(x',\bar \theta)-C(x'-\underline x)
\leq F(\bar x\wedge x',\bar \theta)-C(\bar x\wedge x'-\underline x)
= G(\widehat x,\bar \theta) ,
\end{equation*}
which since $x'$ maximizes $G(\cdot,\bar \theta)$ on $L$ implies that $\widehat x$ does, too.

For the second part, let $\bar x$ be the largest element of $\argmax_{x\in L} F(x,\bar \theta)$, and let $\widehat x\in \argmax_{x\in L} G(x,\bar \theta)$; we will show that $\bar x \geq \widehat x$. The optimality of $\widehat x$ implies that $G(\widehat x,\bar \theta) \geq G( \bar x\wedge \widehat x,\bar \theta)$. It furthermore holds that $C(\widehat x-\underline x)\geq C(\bar x\wedge \widehat x-\underline x)$, by the monotonicity of $C$ and the fact that in each dimension $i$, either $\bar x_i \leq \widehat x_i$ so $0 \leq ( \bar x\wedge \widehat x - \underline x )_i \leq (\widehat x - \underline x )_i$, or $\bar x_i > \widehat x_i$ in which case $( \bar x\wedge \widehat x-\underline x )_i = (\widehat x-\underline x )_i$. Hence $F(\widehat x,\bar \theta) \geq F(\bar x\wedge \widehat x,\bar \theta)$, which implies $F(\bar x\vee \widehat x,\bar \theta) \geq F(\bar x,\bar \theta)$ by quasi-supermodularity. Since $\bar x$ is the largest maximizer of $F(\cdot,\bar \theta)$, it follows that $\bar x \geq \bar x\vee \widehat x$, which is to say that $\bar x \geq \widehat x$.
\end{proof}

The monotonicity assumption in \Cref{theorem:lechatelier} cannot be dropped: if the cost $C$ were merely minimally monotone, then $\bar x \geq \widehat x$ would not necessarily hold.%
    \footnote{For example, if $L = \R$, $F(x,\underline \theta) = -x^2$, $F(x,\bar \theta) = -(x-2)^2$, and $C(\varepsilon) = \infty$ if $0 < \varepsilon < 3$ and $C(\varepsilon) = 0$ otherwise, then $G(\cdot,\bar \theta)$ is uniquely maximized by $\widehat x=3$, and $\bar x = 2 < \widehat x$.}
Monotonicity is not quite necessary, however: a somewhat weaker property is necessary and sufficient for the Le Chatelier principle to hold whatever the objective $F$, as we show in \cref{sec:necessity} below.

%%%%%%%%%%%%%%%%%%%%%%%%%%%%%%%%%%%
\subsection{Extensions}
\label{sec:lechatelier:extensions}
%%%%%%%%%%%%%%%%%%%%%%%%%%%%%%%%%%%

\Cref{theorem:lechatelier} remains true if adjustment is costly also in the long run: that is, if in addition to the short-run cost $C_1(x_1-\underline x)$ of moving from the initial choice $\underline x$ to her short-run choice $x_1$, the agent incurs a further cost $C_2(x_2-x_1)$ of moving from her short-run choice $x_1$ to her long-run choice $x_2$.

\begin{proposition}\label{proposition:lechatelier-medium}
Suppose that the objective $F(x,\theta)$ is quasi-supermodular in $x$ and has single-crossing differences in $(x,\theta)$, and that the adjustment costs $C_1$ and $C_2$ are monotone. If $\bar \theta \geq \underline \theta$, then $x_2 \geq x_1 \geq \underline x$ for some $x_1 \in \argmax_{x \in L} [ F(x,\bar \theta) - C_1(x-\underline x) ]$ and $x_2 \in \argmax_{x \in L} [ F(x,\bar \theta) - C_2(x-x_1) ]$, provided the argmaxes are nonempty.%
\footnote{In fact, $x_1$ and $x_2$ may be chosen so that $\bar x \geq x_2 \geq x_1 \geq \underline x$ holds for any $\bar x \in \argmax_{x \in L} F(x,\bar \theta)$ that satisfies $\bar x \geq \underline x$.}
\end{proposition}

To interpret this result, note that when long-run adjustment is costly, it matters whether or not the agent is forward-looking when making her short-term choice $x_1$, because $x_1$ now enters her long-run payoff $F(x_2,\bar \theta) - C_2(x_2-x_1)$. \Cref{proposition:lechatelier-medium} describes an agent who is myopic, taking no account of the long-run implications of her short-run choice $x_1$. Forward-looking behavior is studied in \cref{sec:dynamic} below. In \cref{sec:myopic}, we revisit myopic behavior, proving a general result (\Cref{theorem:lechatelier-dynamic-myopic}) of which \Cref{proposition:lechatelier-medium} is a special case.

\Cref{theorem:lechatelier} also has a ``$\forall$'' counterpart. Say that the cost function $C$ is \emph{strictly monotone} if and only if $C(\varepsilon') < C(\varepsilon)$ holds whenever $\varepsilon' \neq \varepsilon$ and $\varepsilon'$ is ``between $0$ and $\varepsilon$'' in the sense that in each dimension $i$, either $0 \leq \varepsilon'_i \leq \varepsilon_i$ or $0 \geq \varepsilon'_i \geq \varepsilon_i$. Strict monotonicity implies monotonicity and strict minimal monotonicity.

\begin{proposition} \label{proposition:lechatelier-strict}
Suppose that the objective $F(x,\theta)$ is quasi-supermodular in $x$ and has single-crossing differences in $(x,\theta)$, and that the cost $C$ is \underline{strictly} monotone.
Fix $\bar \theta \geq \underline \theta$, and let $\bar x\in \argmax_{x\in L} F(x,\bar \theta)$ satisfy $\bar x\geq \underline x$.%
\fnsuchan{ }%
Then $\bar x \geq \widehat x \geq \underline x$ for \underline{any} $\widehat x\in \argmax_{x\in L} G(x,\bar \theta)$.
\end{proposition}

%%%%%%%%%%%%%%%%%%%%%%%%%%%%%%%%%%%
\subsection{Application to factor demand}
\label{sec:lechatelier:factor_demand}
%%%%%%%%%%%%%%%%%%%%%%%%%%%%%%%%%%%

Consider a stylized model of production, following \textcite{MilgromRoberts1996}. A firm uses capital $k$ and labor $\ell$ to produce output $f(k,\ell)$. Profit at real factor prices $(r,w)$ is $F(k,\ell,-w) = f(k,\ell) - r k - w \ell$. The adjustment cost $C$ is monotone, but otherwise unrestricted.

If the production function $f$ is supermodular, meaning that capital and labor are complements, then profit $F(k,\ell,-w)$ is supermodular in $x=(k,\ell)$. By inspection, the profit function $F(k,\ell,-w)$ has increasing differences in $(x,\theta) = ((k,\ell),-w)$. So by \Cref{theorem:lechatelier}, any drop in the wage $w$ precipitates a short-run increase of both $k$ and $\ell$, and a further increase in the long run.

If $f$ is instead submodular, meaning that capital and labor are substitutes in production, then we may apply \Cref{theorem:lechatelier} to the choice variable $(x_1,x_2) = (-k,\ell)$, since profit $F^\dag(x_1,x_2,-w) = f(-x_1,x_2) + r x_1 + (- w) x_2$ is then supermodular in $x=(x_1,x_2)$ and has increasing differences in $(x,-w)$.%
\footnote{This trick is due to \textcite{MilgromRoberts1996}.}
The conclusion is that $\ell$ still increases in the short run and further increases in the long run, whereas $k$ now \emph{decreases.}

\textcite{MilgromRoberts1996} were the first to use the theory of monotone comparative statics to obtain such a result. They assumed that labor adjustments are costless and that capital cannot be adjusted at all in the short run: in other words, $C(\varepsilon_k,\varepsilon_\ell) = C_k(\varepsilon_k) + C_\ell(\varepsilon_\ell)$, where $C_\ell \equiv 0$ and $C_k(\varepsilon_k)=\infty$ for every $\varepsilon_k \neq 0$. Our analysis reveals that much weaker assumptions suffice. It turns out not to matter whether labor is cheap to adjust relative to capital. What matters is, rather, that short-run adjustments are costly.

%%%%%%%%%%%%%%%%%%%%%%%%%%%%%%%%%%%
\subsection{Application to pricing}
\label{sec:lechatelier:pricing}
%%%%%%%%%%%%%%%%%%%%%%%%%%%%%%%%%%%

The central plank of new Keynesian macroeconomic models is price stickiness, and the oldest and most important microfoundation for this property is (nonconvex) adjustment costs \parencite[e.g.][]{Mankiw1985,CaplinSpulber1987,GolosovLucas2007,Midrigan2011}. These may be real costs of updating what prices are displayed: empirically, such ``menu costs'' can be nonnegligible \parencite[see e.g.][]{LevyBergenDuttaVenable1997}. Or they may arise from consumers reacting adversely to price hikes by temporarily reducing demand \parencite[as in][]{AnticSalant}.

To study pricing, we consider the simplest model, following \textcite{MilgromRoberts1990}: a monopolist with constant marginal cost $c \geq 0$ faces a decreasing demand curve $D(\cdot,\eta)$ parametrized by $\eta$, thus earning a profit of $F(p,(c,-\eta)) = (p-c) D(p,\eta)$ if she prices at $p \in \R_+$. We assume that demand $D(p,\eta)$ is always strictly positive, and that $\eta$ is an elasticity shifter: when it increases, so does the absolute elasticity of demand at every price $p$. Then profit $F(p,\theta) = F(p,(c,-\eta))$ has increasing differences in $(p,c)$ and has log increasing differences in $(p,-\eta)$, so it has single-crossing differences in $(p,\theta) = (p,(c,-\eta))$. Furthermore, profit $F(p,\theta)$ is automatically quasi-supermodular in $p$ since this choice variable is one-dimensional ($L \subseteq \R$).

Adjusting the price by $\varepsilon$ incurs a cost of $C(\varepsilon) \geq 0$. We assume nothing about $C$ except that it is minimized at zero. In many macroeconomic models, it is a pure fixed cost: $C(\varepsilon) = k>0$ for every $\varepsilon \neq 0$. When adjustment costs arise from price-hike-averse consumers, we have $C(\varepsilon)=0$ for $\varepsilon \leq 0$ and $C(\varepsilon)>0$ for $\varepsilon > 0$. If consumers are inattentive to small price changes, then $C(\varepsilon) = 0$ if $\varepsilon \in \left[ \underline \varepsilon, \bar \varepsilon \right]$ and $C(\varepsilon) > 0$ otherwise, where $\underline \varepsilon < 0 < \bar \varepsilon$.

By \Cref{theorem:basic}, the familiar comparative-statics properties of the monopoly problem are robust to the introduction of adjustment costs: it remains true that the monopolist raises her price whenever her marginal cost $c$ rises and whenever demand becomes less elastic (i.e., $\eta$ falls). No assumptions on the adjustment cost $C$ are required except that it be minimized at zero.

Under the mild additional assumption that $C$ is single-dipped, \Cref{theorem:lechatelier} yields a dynamic prediction: in response to a shock that increases her marginal cost or decreases the elasticity of demand, the monopolist initially raises her price, and then increases it further over the longer run. Thus one-off permanent cost and demand-elasticity shocks lead, quite generally, to price increases in both the short and long run.

A key reason why we can draw such general conclusions about pricing is that \Cref{theorem:basic,theorem:lechatelier} require $F$ to satisfy only ordinal (not cardinal) complementarity conditions. Specifically, we used the fact that the monopolist's profit undergoes a ``single-crossing differences'' shift when demand becomes less elastic (i.e., when $\eta$ falls). A result which assumed the cardinal property of \emph{increasing} differences would have been inapplicable, since elasticity shifts do not generally cause profit to shift in an ``increasing differences'' fashion.%
\footnote{This applied advantage of requiring only ordinal complementarity was pointed out by \textcite{MilgromRoberts1990,MilgromShannon1994} in the context of models with costless adjustment.}

%%%%%%%%%%%%%%%%%%%%%%%%%%%%%%%%%%%
%%%%%%%%%%%%%%%%%%%%%%%%%%%%%%%%%%%
\section{Dynamic adjustment}
\label{sec:dynamic}
%%%%%%%%%%%%%%%%%%%%%%%%%%%%%%%%%%%
%%%%%%%%%%%%%%%%%%%%%%%%%%%%%%%%%%%

The Le Chatelier principle takes a classical, ``reduced-form'' approach to dynamics, following Samuelson and Milgrom--Roberts. In this section and the next, we consider a fully-fledged dynamic model of adjustment. We show that the Le Chatelier principle remains valid: in the short run, the agent's choices exceed the initial choice $\underline x$ and do not overshoot the new frictionless optimum $\bar x$. We furthermore show that under additional assumptions, the path of adjustment is monotone, so that the agent adjusts more over longer horizons.

In this section, we assume that the agent is long-lived and forward-looking. The alternative case in which each period $t$'s choice $x_t$ is made by a short-lived agent (or equivalently, by a myopic long-lived agent) is studied in \cref{sec:myopic}.

%%%%%%%%%%%%%%%%%%%%%%%%%%%%%%%%%%%
\subsection{Setting}
\label{sec:dynamic:setting}
%%%%%%%%%%%%%%%%%%%%%%%%%%%%%%%%%%%

The agent faces an infinite-horizon decision problem in discrete time. In each period $t \in \N = \{1,2,3,\dots\}$, she takes an action $x_t \in L$, and earns a payoff of $F(x_t,\theta_t)$. Adjusting from $x_{t-1}$ to $x_t$ in period $t$ costs $C_t(x_t - x_{t-1})$.

The agent's initial choice $x_0 = \underline x \in \argmax_{x \in L} F(x,\underline \theta)$ is given, as are the parameter sequence $(\theta_t)_{t=1}^\infty$ and the sequence $(C_t)_{t=1}^\infty$ of adjustment cost functions. The simplest example is a one-off parameter shift ($\theta_t=\bar \theta$ for all $t \in \N$) with a time-invariant cost ($C_t = C$ for all $t \in \N$).

The agent is forward-looking, and discounts future payoffs by a factor of $\delta \in (0,1)$. Given her period-$0$ choice $x_0 \in L$, the agent's payoff from a sequence $(x_t)_{t=1}^\infty$ in $L$ is
\begin{gather*}
\mathcal{G}((x_t)_{t=1}^\infty,x_0)
= \mathcal{F}((x_t)_{t=1}^\infty)
-\mathcal{C}(x_0,(x_t)_{t=1}^\infty) ,
\qquad \text{where}
\\
\mathcal{F}((x_t)_{t=1}^{\infty})
= \sum_{t=1}^{\infty} \delta^{t-1}F(x_t,\theta_t)
\quad\text{and}\quad
\mathcal{C}(x_0,(x_t)_{t=1}^{\infty})
= \sum_{t=1}^{\infty} \delta^{t-1}C_t(x_t-x_{t-1}) .
\end{gather*}

\begin{remark} \label{remark:dynamic_extensions_1}
As we describe below (\hyperref[remark:dynamic_extensions_2]{\Cref*{remark:dynamic_extensions_1}, continued}), our results apply also in the finite-horizon case.
\end{remark}

\begin{remark} \label{remark:dynamic_initial_static}
In some applications, the adjustment cost in each period $t$ is calculated relative to the \emph{initial} choice $\underline x$, not (as assumed in this section) relative to the previous period's choice $x_{t-1}$. For example, the initial choice $\underline x$ could be a norm, default or reference point that was formed through custom, bargaining, or other processes. In such cases, the agent incurs a cost in each period of deviating from $\underline x$. The long-lived agent's problem is then the same in every period: choose $x \in L$ to maximize $(1-\delta) \left[ F(x,\bar\theta)-C(x-\underline x) \right]$. This is formally equivalent to the ``static'' model studied in \cref{sec:mcs,sec:lechatelier} above.

A more general model would allow the prevailing norm $y_t$ to evolve sluggishly, for example $y_t = \lambda x_{t-1} + (1-\lambda) \underline x$, where $\lambda \in [0,1]$. (The above discussion concerns $\lambda=1$ and $\lambda=0$.) We do not study this more general model in the present paper, but view it as a potentially interesting avenue for future work.
\end{remark}

%%%%%%%%%%%%%%%%%%%%%%%%%%%%%%%%%%%
\subsection{Dynamic Le Chatelier principles}
\label{sec:dynamic:lechatelier}
%%%%%%%%%%%%%%%%%%%%%%%%%%%%%%%%%%%

The following result shows that our Le Chatelier principle (\Cref{theorem:lechatelier}) remains valid when the agent can adjust over time and is forward-looking: for any new frictionless optimum $\bar x \in\argmax_{x\in L}F(x,\bar\theta)$, the agent's ``short-run'' actions $x_t$ satisfy $\underline x \leq x_t \leq \bar x$ along some optimal path $(x_t)_{t=1}^\infty$.

\begin{theorem}[dynamic Le Chatelier] \label{theorem:lechatelier-dynamic}
Suppose that the objective $F(x,\theta)$ is quasi-supermodular in $x$ and has single-crossing differences in $(x,\theta)$, and that each adjustment cost $C_t$ is monotone. Fix $\bar \theta \geq \underline \theta$, and let $\bar x\in \argmax_{x\in L} F(x,\bar \theta)$ satisfy $\bar x\geq \underline x$.%
\fnsuchan{ }%
If $\underline\theta\leq \theta_t\leq \bar\theta$ for every $t \in \N$, then provided the long-lived agent's problem admits a solution, there is a solution $(x_t)_{t=1}^\infty$ that satisfies $\underline x\leq x_t\leq \bar x$ for every period $t \in \N$.
\end{theorem}

The proof (\cref{sec:appendix:pf_lechatelier-dynamic}) is a direct extension of the arguments used to prove \Cref{theorem:basic,theorem:lechatelier}. The straightforwardness of this extension is perhaps surprising, since the dynamic adjustment problem is superficially quite different from the one-shot problem: the agent chooses a \emph{sequence} of actions, and her objective $\mathcal{F}(\cdot)$ need not be quasi-supermodular (since the sum of quasi-supermodular functions is not quasi-supermodular in general).

Monotonicity cannot be weakened in \Cref{theorem:lechatelier-dynamic}: it is necessary as well as sufficient for the dynamic Le Chatelier principle to hold whatever the objective $F$ and parameter sequence $(\theta_t)_{t=1}^\infty$, as we show in \cref{sec:necessity} below.

The next result shows that under stronger assumptions, a stronger dynamic Le Chatelier principle holds: $\underline x\leq x_t\leq x_T\leq \bar x$ for any periods $t < T$, which is to say that the agent adjusts more at longer horizons. Let us use ``BCS'' as shorthand for ``bounded on compact sets.''

\begin{theorem}[strong dynamic Le Chatelier] \label{theorem:lechatelier-dynamic-strong}
Suppose that the objective $F(x,\theta)$ is supermodular and BCS in $x$ and has single-crossing differences in $(x,\theta)$, and that $C_t = C$ for every period $t$, where the adjustment cost $C$ is monotone and additively separable. Fix $\bar \theta \geq \underline \theta$, and let $\bar x\in \argmax_{x\in L} F(x,\bar \theta)$ satisfy $\bar x\geq \underline x$.%
\fnsuchan{ }%
Let the parameter shift once and for all: $\theta_t = \bar\theta$ for every $t \in \N$. Then provided the long-lived agent's problem admits a solution, there is a solution $(x_t)_{t=1}^\infty$ that satisfies $\underline x \leq x_t\leq x_{t+1}\leq \bar x$ for every period $t \in \N$.
\end{theorem}

The assumptions of \Cref{theorem:lechatelier-dynamic-strong} strengthen those of \Cref{theorem:lechatelier-dynamic} in two directions: (1) the parameter $\theta_t$ and cost function $C_t$ must not vary over time, and (2) the payoff $F$ and cost $C_t=C$ must satisfy additional properties, namely, supermodularity and BCS of $F(\cdot,\theta)$ and additive separability of $C$. The BCS requirement is mild (continuity is a sufficient condition). The other two requirements, supermodularity and additive separability, are substantial assumptions. However, they are both automatically satisfied when the choice variable is one-dimensional ($L \subseteq \R$), as in our applications to pricing, labor supply and capital investment (\cref{sec:dynamic:pricing,sec:dynamic:labor_supply,sec:dynamic:investment} below).

% If adjustment costs are not assumed to be additive across dimensions,
% the result remains valid (with the same proof) provided the cost function $C$
% satisfies $C(0)=0$ and the following property:
% for any vectors $a,b,c \in L$,
% $$C(b\vee c-a\vee b)+C(b\wedge c-a\wedge b)\leq C(b-a)+C(c-b).$$

The proof (\cref{sec:appendix:pf_lechatelier-dynamic-strong}) runs as follows. Any sequence $(x_t)_{t=1}^\infty$ can be made increasing by replacing its $t^\text{th}$ entry $x_t$ with the cumulative maximum $x_1 \vee x_2 \vee \cdots \vee x_{t-1} \vee x_t$, for each $t \in \N$. It suffices to show that such ``monotonization'' preserves optimality, since then the optimal sequence delivered by \Cref{theorem:lechatelier-dynamic} may be monotonized to yield an optimal sequence with all of the desired properties. To prove that monotonization preserves optimality, it suffices (by BCS and a limit argument) to show that an optimal sequence $(x_t)_{t=1}^\infty$ which satisfies $x_1 \leq x_2 \leq \cdots \leq x_{k-1} \leq x_k$ remains optimal if its $t^\text{th}$ entry $x_t$ is replaced by $x_{t-1} \vee x_t$ for each $t \geq k+1$. We prove this using the supermodularity of $F(\cdot,\bar \theta)$ and the monotonicity and additive separability of $C$.

\begin{namedthm}[\Cref*{remark:dynamic_extensions_1}, continued.] \label{remark:dynamic_extensions_2} \upshape
If there is a finite horizon $K$, so the agent chooses a length-$K$ sequence $(x_t)_{t=1}^K$ to maximize $\sum_{t=1}^K \delta^{t-1} \left[ F(x_t,\bar \theta) - C(x_t-x_{t-1}) \right]$, then \Cref{theorem:lechatelier-dynamic} remains directly applicable (set $C_t(\varepsilon) = \infty$ for all $t > K$ and $\varepsilon \neq 0$). \Cref{theorem:lechatelier-dynamic-strong} also remains true as stated (and BCS can be dropped); this is shown in the appendix, immediately after the proof of \Cref{theorem:lechatelier-dynamic-strong}.

% If costs have the general form $\widetilde C_t(x_t,x_{t-1})$, then \Cref{theorem:lechatelier-dynamic} remains true (with the same proof) provided $\varepsilon \mapsto \widetilde C_t(y+\varepsilon,y)$ is monotone for every $t \in \N$ and $y \in L$, and \Cref{theorem:lechatelier-dynamic-strong} remains true (with the same proof) if $\varepsilon \mapsto \widetilde C(y+\varepsilon,y)$ is monotone and additively separable for each $y \in L$ and the cost function $\widetilde C$ is either symmetric ($\widetilde C(x,y) = \widetilde C(y,x)$ for all $x,y \in L$) or translation-invariant ($\widetilde C(x,y) = \widetilde C(x+\varepsilon,y+\varepsilon)$ for all $x,y \in L$ and $\varepsilon \in \Delta L$).
\end{namedthm}

\begin{remark}
There is a familiar way of obtaining comparative statics in dynamic problems, via the Bellman equation \parencite[see][]{HopenhaynPrescott1992}. This approach can be used to obtain an increasing optimal path $(x_t)_{t=1}^\infty$, as in \Cref{theorem:lechatelier-dynamic-strong}, but only under stronger assumptions. The Bellman equation is
\begin{equation*}
V(x)
= \max_{x' \in L} \left[
F( x', \bar \theta ) - C(x'-x) + \delta V(x')
\right]
\quad \text{for every $x \in L$.}
\end{equation*}
Suppose we find conditions which guarantee that
\begin{equation*}
\Phi(x)
= \argmax_{x' \in L}
\left[ F( x', \bar \theta ) 
- C(x'-x) + \delta V(x') \right]
\end{equation*}
is a nonempty compact sublattice for each $x \in L$, and that the correspondence $\Phi$ is increasing in the strong set order. Then $\Phi(x)$ has a greatest element $\phi(x)$ for each $x \in L$, and $\phi$ is an increasing function. Define an optimal sequence $(x_t)_{t=1}^\infty$ by $x_t = \phi(x_{t-1})$ for each $t \in \N$. We have $x_1=\phi (x_0)\geq x_0$ by \Cref{theorem:lechatelier-dynamic}, and by repeatedly applying $\phi$ on both sides of this inequality, we obtain $x_2=\phi(x_1)\geq x_1$, then $x_3=\phi(x_2)\geq x_2$, and so on; thus $(x_t)_{t=1}^\infty$ is increasing.

The usual sufficient conditions for the above argument are that $G(x',x) = F( x', \bar \theta ) - C(x'-x)$ is supermodular in $x'$ and has increasing differences in $(x',x)$. These cardinal complementarity conditions, together with ancillary assumptions, ensure that the value function $V$ is supermodular \parencite[see][]{HopenhaynPrescott1992}, allowing us to apply the basic comparative-statics result \parencite[see][Theorem~4]{MilgromShannon1994} to the objective $G(x',x) + \delta V(x')$ to conclude that $\Phi$ is increasing in the strong set order. The cardinal complementarity conditions effectively require that $-C(x'-x)$ have increasing differences $(x',x)$. This is a very restrictive assumption; in the additively separable case $C(\varepsilon) = \sum_{i=1}^n C_i(\varepsilon_i)$, it demands that each $C_i$ be convex. \Cref{theorem:lechatelier-dynamic-strong} avoids this assumption, requiring merely that the cost $C$ be monotone.
\end{remark}

%%%%%%%%%%%%%%%%%%%%%%%%%%%%%%%%%%%
\subsection{Application to pricing, continued}
\label{sec:dynamic:pricing}
%%%%%%%%%%%%%%%%%%%%%%%%%%%%%%%%%%%

An active literature in macroeconomics \parencite[e.g.][]{GolosovLucas2007,Midrigan2011} examines the price stickiness central to the new Keynesian paradigm by studying forward-looking dynamic models of pricing subject to adjustment costs (usually called ``menu costs'' in this context---see \cref{sec:lechatelier:pricing} above). The basic mechanism is that nonconvexities in adjustment costs give rise to price stickiness.

\Cref{theorem:lechatelier-dynamic-strong} delivers comparative statics for such pricing models, without any of the parametric assumptions that are typically placed on adjustment costs.%
\footnote{Common functional forms include quadratic \parencite{Rotemberg1982} and pure fixed cost \parencite[many papers, e.g.][]{CaplinSpulber1987,GolosovLucas2007,Midrigan2011}.}
Consider again the monopoly pricing problem described in \cref{sec:lechatelier:pricing}. Assume that the adjustment cost $C$ is monotone. It is automatically true that the cost $C$ is additively separable and that profit $F(\cdot,\eta)$ is supermodular, because the choice variable $p \in \R_+$ is one-dimensional. Thus by \Cref{theorem:lechatelier-dynamic-strong}, supply shocks cause inflation at every horizon: a one-off permanent increase of marginal cost $c$ leads prices to increase monotonically over time. The same is true of demand shocks that make the demand curve less elastic (lower $\eta$).

\Cref{theorem:lechatelier-dynamic} furthermore provides that the path of prices remains always above the original frictionless monopoly price, and never overshoots the new frictionless monopoly price. This conclusion is more general, holding even when marginal cost $c_t$, the demand elasticity parameter $\eta_t$, and the adjustment cost function $C_t(\cdot)$ vary over time.

Although we phrased these findings in terms of a monopolist's pricing problem, they apply equally to the typical new Keynesian setting of monopolistic competition between many firms selling differentiated goods \parencite[see e.g.][]{Gali2015}. In that case, the demand curve in our analysis above is to be understood as \emph{residual} demand, taking into account the other firms' pricing.

%%%%%%%%%%%%%%%%%%%%%%%%%%%%%%%%%%%
\subsection{Application to labor supply}
\label{sec:dynamic:labor_supply}
%%%%%%%%%%%%%%%%%%%%%%%%%%%%%%%%%%%

A recent literature attempts to reconcile ``micro'' and ``macro'' estimates of labor supply elasticities by appeal to adjustment costs \parencite[see][]{ChettyEtal2011,Chetty2012}. The idea is that since job design in firms is often inflexible, with no scope for big changes in hours, a worker wishing to adjust her labor supply may have to find a new job, which entails costly search.

The workhorse model of this literature features a worker choosing consumption $c \in \R_+$ and labor supply $x \in L \subseteq \R_+$. Per-period utility has the quasilinear form $c - \kappa(x)$. Given quasilinearity, it is without loss of generality to assume that the worker cannot save or borrow.%
\footnote{Given the standard assumption that the worker's discount rate is equal to the interest rate, the marginal utility of a dollar is equal across all periods.}
She then consumes whatever she earns, so her per-period utility equals $F(x,T) = wx - T(wx) - \kappa(x)$, where $w > 0$ is the wage and $T(\cdot)$ is the tax schedule, and we have normalized the price of consumption to $1$. Labor supply is subject to adjustment costs (arising from search), which are assumed to be single-dipped. We depart from the literature by eschewing functional-form restrictions on the effort disutility $\kappa$, tax schedule $T$ and adjustment cost function.

We consider tax reforms which reduce marginal rates. Formally, we write $\bar T \geq_{\text{flat}} T$ (``$\bar T$ is flatter than $T$'') if and only if $\bar T(y') - \bar T(y) \leq T(y') - T(y)$ for all $y' \geq y \geq 0$. The per-period utility $F(x,T)$ has single-crossing differences in $(x,T)$,%
\footnote{For $x' \geq x$ and $\bar T \geq_{\text{flat}} T$, $F(x',T) - F(x,T) \geq \mathrel{(>)} 0$ is equivalent to $T(wx') - T(wx) \leq \mathrel{(<)} wx'-wx - k$ where $k = \kappa(x')-\kappa(x)$, which implies $\bar T(wx') - \bar T(wx) \leq \mathrel{(<)} wx'-wx - k$ by definition of $\geq_{\text{flat}}$, which is equivalent to $F(x',\bar T) - F(x,\bar T) \geq \mathrel{(>)} 0$.}
and is automatically supermodular in $x$.

Our Le Chatelier principles imply that when marginal tax rates are cut, labor supply increases at every horizon. More specifically, in the ``classical'' environment of \cref{sec:setting,sec:mcs,sec:lechatelier}, labor supply increases in the short run and further increases in the long run (\Cref{theorem:lechatelier}), while in the infinite-horizon model of \cref{sec:dynamic:setting,sec:dynamic:lechatelier}, labor supply rises monotonically over time (\Cref{theorem:lechatelier-dynamic-strong}). The adjustment may be gradual or abrupt, depending on functional forms.

%%%%%%%%%%%%%%%%%%%%%%%%%%%%%%%%%%%
\subsection{Application to capital investment}
\label{sec:dynamic:investment}
%%%%%%%%%%%%%%%%%%%%%%%%%%%%%%%%%%%

In the neoclassical theory of investment \parencite[originating with][]{Jorgenson1963}, a firm adjusts its capital stock over time subject to adjustment costs. In the simplest such model, the profit of a firm with capital stock $k_t \in \R_+$ is $F(k_t,(p,\eta,-r)) = p f(k_t,\eta) - r k_t$, where $(p,r)$ are the prices of output and capital and $f(\cdot,\eta)$ is an increasing production function. Capital is subject to an adjustment cost: investing $i_t = k_t - k_{t-1}$ costs $C(i_t) \geq 0$, where $C(0)=0$.

We assume that $f$ has increasing differences, so that the parameter $\eta$ shifts the marginal product of capital. Then $F(k,\theta)$ has increasing differences (and hence single-crossing differences) in $(k,\theta)$, where $\theta = (p,\eta,-r)$. Profit $F(k,\theta)$ is automatically supermodular in $k$, since $k \in \R_+$ is one-dimensional. Our discussion below may be extended to richer variants of this model featuring, for example, depreciation and time-varying prices.

The early literature assumed a convex adjustment cost $C(\cdot)$, which yields gradual capital accumulation and an equivalence of the neoclassical theory with Tobin's (\citeyear{Tobin1969}) ``$q$'' theory of investment \parencite[see][]{Hayashi1982}. Later work focused on the ``lumpy'' investment behavior that arises when adjustment costs are nonconvex. ``Lumpiness'' is empirically well-documented \parencite[see][]{CooperHaltiwanger2006}, and has implications for, among other things, business cycles \parencite[e.g.][]{Thomas2002,BachmannCaballeroEngel2013,Winberry2021} and the effects of microfinance programs on entrepreneurship in developing countries \parencite[e.g.][]{FieldPandePappRigol2013,BariMalikMekiQuinn2024}.

Our comparative-statics theory handles both the convex case and rich forms of nonconvexity. Our Le Chatelier principles (the ``classical,'' reduced-form \Cref{theorem:lechatelier} and the dynamic, forward-looking \Cref{theorem:lechatelier-dynamic,theorem:lechatelier-dynamic-strong}) are applicable provided merely that adjustment costs are single-dipped. Investment then increases at every horizon, and by more at longer horizons, whenever the marginal profitability of capital increases, whether due to a drop in its price $r$, a rise in the price $p$ of output, or an increase of the marginal product of capital (an increase of $\eta$).

The aforementioned papers on microfinance consider models in which adjustment costs fail even to be single-dipped: there is a minimum investment size $I>0$, meaning that investing $i \in (0,I)$ costs $C(i)=\infty$ (whereas investing $i \geq I$ has finite cost). Our fundamental result, \Cref{theorem:basic}, can accommodate such failures of single-dippedness: it remains true that a rise in the marginal profitability of capital increases investment, just as would be the case if adjustment were costless. Our remaining results (\Cref{theorem:lechatelier,theorem:lechatelier-dynamic,theorem:lechatelier-dynamic-strong}) cannot be applied in this case, however.

All of these results generalize to multiple factors of production, on the pattern of \cref{sec:lechatelier:factor_demand}. It suffices to assume that the factors $x=(x_1,\dots,x_n)$ are complements in production, meaning that the production function $f(x,\eta)$ is supermodular in $x$. Then, denoting factor prices by $r = (r_1,\dots,r_n)$, the profit function $F(x,(p,\eta,-r)) = p f(x,\eta) - r \cdot x$ is supermodular in $x$, and has increasing differences in $(x,(p,\eta,-r))$ as before, so that all of our general results remain applicable. In case there are just $n=2$ factors of production, the complementarity hypothesis may be replaced with substitutability (submodularity of $f(x,\eta)$ in $x$), using the trick described in \cref{sec:lechatelier:factor_demand}.

%%%%%%%%%%%%%%%%%%%%%%%%%%%%%%%%%%%
%%%%%%%%%%%%%%%%%%%%%%%%%%%%%%%%%%%
\section{Dynamic adjustment by short-lived agents}
\label{sec:myopic}
%%%%%%%%%%%%%%%%%%%%%%%%%%%%%%%%%%%
%%%%%%%%%%%%%%%%%%%%%%%%%%%%%%%%%%%

In this section, we continue our study of the dynamic adjustment model introduced in the previous section, under a different behavioral assumption: that each period's decision is made by a short-lived agent. We recover the Le Chatelier principle, and provide conditions under which short-lived agents adjust more sluggishly than a long-lived agent would.

Recall the model from \cref{sec:dynamic:setting}. We now assume that there is one agent per period. The period-$t$ agent takes her predecessor's choice $x_{t-1} \in L$ as given, and chooses $x_t \in L$ to maximize the period-$t$ payoff $G_t(\cdot,x_{t-1})$, where $G_t(x,y) = F(x,\theta_t) - C_t(x-y)$. A \emph{(short-lived) equilibrium sequence} is a sequence $(x_t)_{t=1}^\infty$ in $L$ such that $x_t \in \argmax_{x \in L} G_t(x,x_{t-1})$ for every $t \in \N$, where (as before) $x_0 = \underline x \in \argmax_{x \in L} F(x,\underline \theta)$ is given.

\begin{theorem}[short-lived dynamic Le Chatelier]\label{theorem:lechatelier-dynamic-myopic}
Suppose that the objective $F(x,\theta)$ is quasi-supermodular in $x$ and has single-crossing differences in $(x,\theta)$, and that each adjustment cost $C_t$ is monotone. Assume that $\argmax_{x \in L}G_t(x,y)$ is nonempty for all $t \in \N$ and $y\in L$. Fix $\bar \theta \geq \underline \theta$, and let $\bar x\in \argmax_{x\in L} F(x,\bar \theta)$ satisfy $\bar x\geq \underline x$.%
\fnsuchan{ }%
\begin{itemize}
\item If $\underline\theta\leq \theta_t\leq \bar\theta$ for every $t \in \N$, then there is an equilibrium sequence $(x_t)_{t=1}^\infty$ that satisfies $\underline x\leq x_t\leq \bar x$ for every $t \in \N$.
\item If $\underline \theta \leq \theta_t \leq \theta_{t+1} \leq \bar \theta$ for every $t \in \N$, then there is an equilibrium sequence $(x_t)_{t=1}^\infty$ that satisfies $\underline x \leq x_t\leq x_{t+1} \leq \bar x$ for every $t \in \N$.
\end{itemize}
\end{theorem}

By inspection, \Cref{proposition:lechatelier-medium} (\cref{sec:lechatelier}) is the special case of \Cref{theorem:lechatelier-dynamic-myopic} in which there are only two periods of adjustment ($C_t(\varepsilon) = \infty$ for all $t \geq 3$ and $\varepsilon \neq 0$) and the parameter shifts once and for all ($\theta_1=\theta_2=\bar\theta$).

Taken together, our three dynamic Le Chatelier principles (\Cref{theorem:lechatelier-dynamic,theorem:lechatelier-dynamic-strong,theorem:lechatelier-dynamic-myopic}) tell us that long- and short-lived agents adjust in the same direction. The \emph{speed} of adjustment generally differs, however. The following result gives additional assumptions under which short-lived agents adjust more sluggishly.

\begin{theorem}[short- vs. long-lived] \label{theorem:myopic-vs-fwd}
Suppose that the objective $F(x,\theta)$ is supermodular in $x$ and has single-crossing differences in $(x,\theta)$, and that each adjustment cost $C_t$ is monotone, additively separable and convex. Let $\{ F(\cdot,\theta) \}_{\theta \in \Theta}$ and $\{ C_t \}_{t \in \N}$ be equi-BCS.%
\footnote{Given $X \subseteq \R^n$, a collection $\{ \phi_k \}_{k \in \mathcal{K}}$ of functions $\phi_k : X \to (-\infty,\infty]$ is \emph{equi-BCS} if and only if the map $x \mapsto \sup_{k \in \mathcal{K}} \lvert \phi_k(x) \rvert$ is BCS.}
Fix $\bar \theta \geq \underline \theta$, let $\bar x\in \argmax_{x\in L} F(x,\bar \theta)$ satisfy $\bar x\geq \underline x$,%
\fnsuchan{ }%
and assume that $\underline \theta \leq \theta_t \leq \theta_{t+1} \leq \bar \theta$ for every $t \in \N$. Fix a short-lived equilibrium sequence $(\widetilde x_t)_{t=1}^\infty$ that satisfies $\underline x \leq \widetilde x_t \leq \widetilde x_{t+1} \leq \bar x$ for every $t \in \N$.%
\footnote{Such a sequence exists by \Cref{theorem:lechatelier-dynamic-myopic}, provided there is an equilibrium sequence.}
Then provided the long-lived agent's problem admits a solution, there is a solution $(x_t)_{t=1}^\infty$ that satisfies $\widetilde x_t\leq x_t\leq \bar x$ for every period $t \in \N$.
\end{theorem}

The intuition is straightforward: short-lived agents adjust more sluggishly because they do not take into account that the less they adjust today, the more adjustment will be required tomorrow. Note that this result relies on stronger assumptions than our earlier results: in particular, convexity of costs.

The proof of \Cref{theorem:myopic-vs-fwd} (\cref{sec:appendix:pf_lechatelier-myopic-vs-fwd}) establishes that if $(x_t)_{t=1}^\infty$ solves the long-lived agent's problem and satisfies $\underline x \leq x_t \leq \bar x$ for every $t \in \N$,%
    \footnote{Such a solution exists by \Cref{theorem:lechatelier-dynamic}, provided there is a solution.}
then $( \widetilde x_t \vee x_t )_{t=1}^\infty$ also solves the long-lived agent's problem. By equi-BCS and a limit argument, it suffices to show that for every $T \in \{0,1,2,\dots\}$, the sequence $( \widetilde x_{\min\{t,T\}} \vee x_t )_{t=1}^\infty$ solves the long-lived agent's problem, where $\widetilde x_0 = \underline x$. We show this by induction on $T \in \{0,1,2,\dots\}$; the base case $T=0$ is immediate, and the induction step uses the supermodularity of the objective $F(\cdot,\theta)$ and the monotonicity, additive separability and convexity of each cost $C_t$.

%%%%%%%%%%%%%%%%%%%%%%%%%%%%%%%%%%%
%%%%%%%%%%%%%%%%%%%%%%%%%%%%%%%%%%%
\section{Necessity of monotonicity assumptions}
\label{sec:necessity}
%%%%%%%%%%%%%%%%%%%%%%%%%%%%%%%%%%%
%%%%%%%%%%%%%%%%%%%%%%%%%%%%%%%%%%%

The basic comparative-statics result for costless adjustment \parencite[see][Theorem~4]{MilgromShannon1994} includes a converse, which asserts that the ordinal complementarity assumptions on the objective are not only sufficient for comparative statics, but also necessary in a natural sense. When adjustment is costly, ordinal complementarity is still necessary.%
\footnote{Since ordinal complementarity is necessary for comparative statics in the costless case $C \equiv 0$, it is \emph{a fortiori} necessary for comparative statics to hold whatever the cost $C$.}
In this section, we show that the same is true of monotonicity-type assumptions on the cost function: they are necessary (in the same sense) for comparative statics.%
\footnote{Thanks to a referee and the editor for encouraging us to develop this material.}

Firstly, \Cref{theorem:basic} has a direct converse:

\begin{namedthm}[\Cref*{theorem:basic}$\boldsymbol{^\dag}$.]\label{proposition:basic_charac}
Assume that there exist distinct parameters $\underline \theta, \bar \theta \in \Theta$ such that $\underline \theta \leq \bar \theta$. For a cost function $C$, the following are equivalent:
\begin{enumerate}[label=(\alph*)]
\item \label{item:basic_charac:minmon} $C$ is minimally monotone.
\item \label{item:basic_charac:mcs} For any sublattice $X \subseteq L$ and any objective $F : X \times \Theta \to \R$ such that $F(x,\theta)$ is quasi-supermodular in $x$ and has single-crossing differences in $(x,\theta)$,
\begin{itemize}
\item $\bar \theta \geq \underline \theta$ implies $\widehat x\geq \underline x$ for some $\widehat x\in\argmax_{x\in X} G(x,\bar \theta)$, provided the argmax is nonempty, and
\item $\bar \theta \leq \underline \theta$ implies $\widehat x\leq \underline x$ for some $\widehat x\in\argmax_{x\in X} G(x,\bar \theta)$, provided the argmax is nonempty.
\end{itemize}
\end{enumerate}
\end{namedthm}

Minimal monotonicity is \emph{not} equivalent to the ``$\bar \theta \geq \underline \theta$'' half of property~\ref{item:basic_charac:mcs} alone. Rather, as noted in \cref{footnote:one-half}, the ``$\bar \theta \geq \underline \theta$'' half of \ref{item:basic_charac:mcs} is implied by one half of minimal monotonicity ($C( \varepsilon \vee 0 ) \leq C(\varepsilon)$ for every $\varepsilon \in \Delta L$), while the ``$\bar \theta \leq \underline \theta$'' half follows from the other half of minimal monotonicity.

\begin{proof}[\normalfont\bfseries Proof]
\ref{item:basic_charac:minmon} implies \ref{item:basic_charac:mcs} by \Cref{theorem:basic}. To show that \ref{item:basic_charac:mcs} implies \ref{item:basic_charac:minmon}, we prove the contrapositive: suppose that $C$ fails to be minimally monotone, meaning that there is an adjustment vector $\varepsilon \in \Delta L$ such that $C( \varepsilon \wedge 0 ) > C( \varepsilon )$ or $C( \varepsilon \vee 0 ) > C( \varepsilon )$; we will show that \ref{item:basic_charac:mcs} fails. Assume that $C( \varepsilon \vee 0 ) > C( \varepsilon )$; the other case is analogous. Let $\alpha = C( \varepsilon \vee 0 ) - C( \varepsilon ) > 0$ and $\beta = 2 \alpha + \max\{ 0, C( \varepsilon ) - C( \varepsilon \wedge 0 ), C( \varepsilon ) - C( 0 ) \}$.

Choose $\underline x, \widehat x \in L$ such that $\widehat x - \underline x = \varepsilon$, and note that $\widehat x \ngeq \underline x$ (else $\varepsilon \vee 0 = \varepsilon$, which would imply $\alpha=0$). Let $X = \{ \underline x \wedge \widehat x, \underline x, \widehat x, \underline x \vee \widehat x \}$; $X$ is a sublattice. Fix distinct $\underline \theta \leq \bar \theta$ in $\Theta$ (possible by hypothesis). Let $F( \underline x \wedge \widehat x,\underline \theta) = F( \underline x \wedge \widehat x,\bar \theta) = \alpha/2$, $F( \underline x,\underline \theta) = F( \underline x,\bar \theta) = \alpha$, $F( \widehat x,\underline \theta) = 0$ and $F( \widehat x,\bar \theta) = \beta$, $F( \underline x \vee \widehat x,\underline \theta) = \alpha/2$ and $F( \underline x \vee \widehat x,\bar \theta) = \alpha/2+\beta$.%
\footnote{For $\theta \in \Theta \setminus \{ \underline \theta, \bar \theta \}$, let $F(\cdot,\theta) = F(\cdot,\underline \theta)$ if $\theta \leq \underline \theta$ and $F(\cdot,\theta) = F(\cdot,\bar \theta)$ otherwise.}
Then $F(x,\theta)$ is supermodular in $x$ and has increasing differences in $(x,\theta)$, and we have
\begin{equation*}
\argmax_{x \in X} F(x,\underline \theta) = \{ \underline x \}
\quad \text{and} \quad
\argmax_{x \in X} G(x,\bar \theta) = \{ \widehat x \} .
\end{equation*}
Hence \ref{item:basic_charac:mcs} fails, as $\bar \theta \geq \underline \theta$ and $\widehat x \ngeq \underline x$.
\end{proof}

The converse of \Cref{theorem:lechatelier} is not true; instead, a property slightly weaker than monotonicity is necessary and sufficient for the Le Chatelier principle. Say that $C : \Delta L \to [0,\infty]$ is \emph{weakly monotone} if and only if $C(\varepsilon') \leq C(\varepsilon)$ whenever $0 \leq \varepsilon' \leq \varepsilon \vee 0$ or $0 \geq \varepsilon' \geq \varepsilon \wedge 0$. Equivalently, $C$ is weakly monotone if and only if it is minimally monotone and $C(\varepsilon') \leq C(\varepsilon)$ whenever $0 \leq \varepsilon' \leq \varepsilon$ or $0 \geq \varepsilon' \geq \varepsilon$. The latter requirement is monotonicity, except applied only to ``unidirectional'' adjustments (up in every dimension, $\varepsilon \geq 0$, or down in every dimension, $\varepsilon \leq 0$).

\begin{namedthm}[\Cref*{theorem:lechatelier}$\boldsymbol{^\dag}$.]\label{proposition:lechatelier_charac}
Assume that there exist distinct parameters $\underline \theta, \bar \theta \in \Theta$ such that $\underline \theta \leq \bar \theta$. For a cost function $C$, the following are equivalent:
\begin{enumerate}[label=(\alph*)]
\item \label{item:lechatelier_charac:weakmon} The adjustment cost $C$ is weakly monotone.
\item \label{item:lechatelier_charac:mcs} For any sublattice $X \subseteq L$ and any objective $F : X \times \Theta \to \R$ such that $F(x,\theta)$ is quasi-supermodular in $x$ and has single-crossing differences in $(x,\theta)$,
\begin{itemize}
\item $\bar \theta \geq \underline \theta$ implies that for any $\bar x \in \argmax_{x \in X} F(x,\bar \theta)$ such that $\bar x \geq \underline x$,%
\fnsuchan{ }%
we have $\bar x \geq \widehat x \geq \underline x$ for some $\widehat x \in \argmax_{x \in X} G(x,\bar \theta)$, provided the argmax is nonempty, and
\item $\bar \theta \leq \underline \theta$ implies that for any $\bar x \in \argmax_{x \in X} F(x,\bar \theta)$ such that $\bar x \leq \underline x$,%
\fnsuchan{ }%
we have $\bar x \leq \widehat x \leq \underline x$ for some $\widehat x \in \argmax_{x \in X} G(x,\bar \theta)$, provided the argmax is nonempty.
\end{itemize}
\end{enumerate}
\end{namedthm}

\Cref{theorem:lechatelier-dynamic} also has a converse: monotonicity is necessary as well as sufficient for the dynamic Le Chatelier principle. Define $\underline C, \bar C : \Delta L \to [0,\infty]$ by $\underline C \equiv 0$, $\bar C(0)=0$, and $\bar C(\varepsilon) = \infty$ for all $\varepsilon \neq 0$, and call a set $\mathcal{C}$ of cost functions \emph{admissible} if and only if it contains both $\underline C$ and $\bar C$.

\begin{namedthm}[\Cref*{theorem:lechatelier-dynamic}$\boldsymbol{^\dag}$.]\label{proposition:lechatelier-dynamic_charac}
Assume that there exist distinct parameters $\underline \theta, \theta^1, \theta^2, \bar \theta \in \Theta$ such that $\underline \theta \leq \theta^1 \leq \bar \theta \geq \theta^2 \geq \underline \theta$ and $\theta^1 \nleq \theta^2 \nleq \theta^1$. For an admissible set $\mathcal{C}$ of cost functions, the following are equivalent:
\begin{enumerate}[label=(\alph*)]
\item \label{item:lechatelier-dynamic_charac:mon} Every member of $\mathcal{C}$ is monotone.
\item \label{item:lechatelier-dynamic_charac:mcs} For any sublattice $X \subseteq L$, any objective $F : X \times \Theta \to \R$ such that $F(x,\theta)$ is quasi-supermodular in $x$ and has single-crossing differences in $(x,\theta)$, any parameter sequence $(\theta_t)_{t=1}^\infty \subseteq \Theta$, and any cost-function sequence $(C_t)_{t=1}^\infty \subseteq \mathcal{C}$, it holds that
\begin{itemize}
\item for any $\underline \theta \leq \bar \theta$ in $\Theta$ and any $\underline x\in \argmax_{x\in X} F(x,\underline \theta)$ and $\bar x\in \argmax_{x\in X} F(x,\bar \theta)$ such that $\underline x \leq \bar x$,%
\fnsuchmult{ }%
if $\underline\theta\leq \theta_t\leq \bar\theta$ for every $t \in \N$, then provided the long-lived agent's problem admits a solution, there is a solution $(x_t)_{t=1}^\infty$ that satisfies $\underline x\leq x_t\leq \bar x$ for every period $t \in \N$, and
\item for any $\underline \theta \geq \bar \theta$ in $\Theta$ and any $\underline x\in \argmax_{x\in X} F(x,\underline \theta)$ and $\bar x\in \argmax_{x\in X} F(x,\bar \theta)$ such that $\underline x \geq \bar x$,%
\fnsuchmult{ }%
if $\underline\theta\geq \theta_t\geq \bar\theta$ for every $t \in \N$, then provided the long-lived agent's problem admits a solution, there is a solution $(x_t)_{t=1}^\infty$ that satisfies $\underline x\geq x_t\geq \bar x$ for every period $t \in \N$.
\end{itemize}
\end{enumerate}
\end{namedthm}

The role of the assumption that there exist incomparable parameters ($\theta^1,\theta^2 \in \Theta$ such that $\theta^1 \nleq \theta^2 \nleq \theta^2$) is to avoid conflating property~\ref{item:lechatelier-dynamic_charac:mcs} with the weaker property that quantifies only over ``one-dimensional'' parameter sequences $(\theta_t)_{t=1}^\infty$, meaning those such that either $\theta_t \leq \theta_T$ or $\theta_T \leq \theta_t$ holds for all $t,T \in \N$. We do not know whether \ref{item:lechatelier-dynamic_charac:mon} is necessary also for the weaker ``one-dimensional'' property. However, \emph{weak} monotonicity is necessary, by \hyperref[proposition:lechatelier_charac]{\Cref*{theorem:lechatelier}$^\dag$} (see \cref{footnote:thm23_necessity} below).

The proof of \hyperref[proposition:lechatelier-dynamic_charac]{\Cref*{theorem:lechatelier-dynamic}$^\dag$} (\cref{sec:appendix:necessity:pf_lechatelier-dynamic_charac}) has three parts. First, \ref{item:lechatelier-dynamic_charac:mon} implies \ref{item:lechatelier-dynamic_charac:mcs} by \Cref{theorem:lechatelier-dynamic}. Second, if there is a $C \in \mathcal{C}$ that fails to be \emph{weakly} monotone, then \ref{item:lechatelier-dynamic_charac:mcs} fails by \hyperref[proposition:lechatelier_charac]{\Cref*{theorem:lechatelier}$^\dag$}.%
    \footnote{\label{footnote:thm23_necessity}By \hyperref[proposition:lechatelier_charac]{\Cref*{theorem:lechatelier}$^\dag$}, there exist suitable $X \subseteq L$ and $\widetilde F : X \times \Theta \to \R$ such that every $\widehat x \in M = \argmax_{x \in X} [ \widetilde F(x,\bar \theta) - C(x-\underline x) ]$ fails to lie between $\underline x$ and $\bar x$. Hence if $F = (1-\delta) \widetilde F$, $\theta_t = \bar \theta$ for every $t \in \N$, $C_1 = C$ and $C_t = \bar C$ for every $t \geq 2$, then every solution $(x_t)_{t=1}^\infty$ of the long-lived agent's problem has $x_1 \in M$, so $x_1$ fails to lie between $\underline x$ and $\bar x$.}
Third, suppose there is a $C \in \mathcal{C}$ that is weakly monotone but not monotone. In this case, we construct suitable $X \subseteq L$ and $F : X \times \Theta \to \R$ such that if $\theta_1=\theta^1$, $\theta_t=\theta^2$ for all $t \geq 2$, $C_1=\underline C$, $C_2 = C$ and $C_t = \bar C$ for every $t \geq 3$, then in every solution $(x_t)_{t=1}^\infty$ of the long-lived agent's problem, $x_2$ fails to lie between $\underline x$ and $\bar x$.

We do not have a converse for \Cref{theorem:lechatelier-dynamic-strong}. We conjecture that the strong dynamic Le Chatelier principle remains true under somewhat weaker assumptions, but showing this would seem to require a different proof strategy, since our strategy relies heavily on the stated assumptions. Identifying the necessary and sufficient conditions seems challenging; we leave this problem open.

\crefalias{section}{appsec}
\crefalias{subsection}{appsec}
\crefalias{subsubsection}{appsec}

\renewcommand*{\thesection}{Appendix}
%%%%%%%%%%%%%%%%%%%%%%%%%%%%%%%%%%%
%%%%%%%%%%%%%%%%%%%%%%%%%%%%%%%%%%%
\section[Appendix]{}
\label{sec:appendix}
%%%%%%%%%%%%%%%%%%%%%%%%%%%%%%%%%%%
%%%%%%%%%%%%%%%%%%%%%%%%%%%%%%%%%%%

\renewcommand*{\thesubsection}{\Alph{subsection}}
\setcounter{subsection}{0}

%%%%%%%%%%%%%%%%%%%%%%%%%%%%%%%%%%%
\subsection{Standard definitions}
\label{sec:appendix:definitions}
%%%%%%%%%%%%%%%%%%%%%%%%%%%%%%%%%%%

Given $X \subseteq \R$, a function $\psi : X \to (-\infty,\infty]$ is \emph{single-dipped} if and only if there is an $x \in X$ such that $\psi$ is decreasing on $\{ y \in X : y \leq x \}$ and increasing on $\{ y \in X : y \geq x \}$.
Given $X \subseteq \R^n$, a function $\psi : X \to (-\infty,\infty]$ is \emph{bounded on compact sets (BCS)} if and only if for each compact $Y \subseteq X$, there is a constant $K_Y > 0$ such that $\lvert \psi(y) \rvert \leq K_Y$ for every $y \in Y$.
% and a collection $\{ \phi_k \}_{k \in \mathcal{K}}$ of functions $\phi_k : X \to (-\infty,\infty]$ is \emph{equi-BCS} if and only if the map $x \mapsto \sup_{k \in \mathcal{K}} \lvert \phi_k \rvert$ is BCS.

% For $x,y \in \R^n$,
% we write
% $x \wedge y
% = ( \min\{x_1,y_1\}, \dots, \min\{x_n,y_n\} )$
% and
% $x \vee y
% = ( \max\{x_1,y_1\}, \dots, \max\{x_n,y_n\} )$.
% A set $L \subseteq \R^n$ is a \emph{sublattice} of $\R^n$
% if and only if for any $x,y \in L$,
% the vectors $x \wedge y$ and $x \vee y$ also belong to $L$.

% Similarly, for any nonempty set $X \subseteq \R^n$,
% we write
% \begin{equation*}
% \bigwedge X
% = \left( \inf_{x \in X} x_1, \dots, \inf_{x \in X} x_n \right)
% \quad \text{and} \quad
% \bigvee X
% = \left( \sup_{x \in X} x_1, \dots, \sup_{x \in X} x_n \right) .
% \end{equation*}
% A set $L \subseteq \R^n$ is a \emph{subcomplete sublattice} of $\R^n$ if for every nonempty $X \subseteq L$, the vectors $\bigwedge X$ and $\bigvee X$ belong to $L$.
% If a nonempty set $X \subseteq \R^n$ contains $\bigvee X$,
% we call $\bigvee X$ the \emph{greatest element} of $X$, and denote it by $\max X$.
% Similarly for the \emph{least element,} denoted $\min X$.

Fix a sublattice $L$ of $\R^n$. A function $\phi : L \to \R$ is called \emph{supermodular} if $\phi(x) - \phi(x \wedge y) \leq \phi(x \vee y) - \phi(y)$ for any $x,y \in L$, \emph{quasi-supermodular} if $\phi(x) - \phi(x \wedge y) \geq \mathrel{(>)} 0$ implies $\phi(x \vee y) - \phi(y) \geq \mathrel{(>)} 0$, and \emph{(quasi-)submodular} if $-\phi$ is (quasi-)supermodular. Clearly supermodularity implies quasi-supermodularity. If $n=1$, then every function $\phi : L \to \R$ is automatically supermodular.

Fix a partially ordered set $\Theta$. A function $F : L \times \Theta \to \R$ has \emph{(strict) increasing differences} if $F(y,\theta)-F(x,\theta)$ is (strictly) increasing in $\theta$ whenever $x \leq y$, has \emph{single-crossing differences} if $F(y,\theta')-F(x,\theta') \geq \mathrel{(>)} 0$ implies $F(y,\theta'')-F(x,\theta'') \geq \mathrel{(>)} 0$ whenever $x \leq y$ and $\theta' \leq \theta''$, has \emph{strict single-crossing differences} if
$F(y,\theta')-F(x,\theta') \geq 0$ implies $F(y,\theta'')-F(x,\theta'') > 0$ whenever $x < y$ and $\theta' < \theta''$, and has \emph{(strict) decreasing differences} if $-F$ has (strict) increasing differences. A function $F : L \times \Theta \to \R_{++}$ has \emph{(strict) log increasing differences} if and only if $\ln F$ has (strict) increasing differences. (Strict) increasing differences and (strict) log increasing differences each imply (strict) single-crossing differences.

Quasi-supermodularity and single-crossing differences are \emph{ordinal} properties: they are preserved by strictly increasing transformations.%
    \footnote{That is, for any strictly increasing $f : \R \to \R$, if $\phi$ is quasi-supermodular, then so is $f \circ \phi$, and if $F(x,\theta)$ has increasing differences in $(x,\theta)$, then so does $\widetilde F(x,\theta) = f( F(x,\theta) )$.}
By contrast, supermodularity and increasing differences are in general preserved only by strictly increasing \emph{affine} transformations: in other words, they are \emph{cardinal.}

The sum of quasi-supermodular functions need not be quasi-supermodular. Likewise, single-crossing differences is not preserved by summation. By contrast, the sum of supermodular functions is supermodular, and the sum of functions with increasing differences also has increasing differences.

%%%%%%%%%%%%%%%%%%%%%%%%%%%%%%%%%%%
\subsection{Extension: uncertain adjustment cost}
\label{sec:appendix:uncertain}
%%%%%%%%%%%%%%%%%%%%%%%%%%%%%%%%%%%

Several of our results are robust to uncertainty about adjustment costs. For \Cref{theorem:basic,theorem:lechatelier}, augment the setting from \cref{sec:setting} as follows. Let all uncertainty be summarized by a random variable $S$, called ``the state of the world.'' The agent's adjustment cost is $C_s(\cdot)$ in state $S=s$. Her ex-ante payoff is
\begin{equation*}
\widetilde G(x,\theta)
= \E\bigl[ u\bigl(
F(x,\theta) - C_S(x-\underline x)
\bigr) \bigr] ,
\end{equation*}
where $u$ is an increasing function $\R \to \R$.

\begin{namedthm}[\Cref*{theorem:basic}$\boldsymbol{'}$.]\label{theorem:basic-uncertain}
Suppose that the objective $F(x,\theta)$ is quasi-supermodular in $x$ and has single-crossing differences in $(x,\theta)$, and that at almost every realization $s$ of the state $S$, the adjustment cost $C_s$ is minimally monotone. If $\bar \theta \geq \underline \theta$, then $\widehat x\geq \underline x$ for some $\widehat x\in\argmax_{x\in L} \widetilde G(x,\bar \theta)$, provided the argmax is nonempty.
\end{namedthm}

\begin{proof}[\normalfont\bfseries Proof]
Let $x'\in \argmax_{x\in L} \widetilde G(x,\bar \theta)$, and let $\widehat x = \underline x\vee x'$. For almost every realization $s$, applying the proof of \Cref{theorem:basic} yields $F(\widehat x,\bar \theta)-C_s(\widehat x-\underline x)
\geq F(x',\bar \theta)-C_s(x'-\underline x)$. Hence $\widetilde G(\widehat x,\bar \theta) \geq \widetilde G(x',\bar \theta)$ as $u$ is increasing, which since $x'$ maximizes $\widetilde G(\cdot,\bar \theta)$ on $L$ implies that $\widehat x$ does, too. Clearly $\widehat x \geq \underline x$.
\end{proof}

\begin{namedthm}[\Cref*{theorem:lechatelier}$\boldsymbol{'}$.]\label{theorem:lechatelier-uncertain}
Suppose that the objective $F(x,\theta)$ is quasi-supermodular in $x$ and has single-crossing differences in $(x,\theta)$, and that at almost every realization $s$ of the state $S$, the adjustment cost $C_s$ is monotone. Fix $\bar \theta \geq \underline \theta$, and let $\bar x\in \argmax_{x\in L} F(x,\bar \theta)$ satisfy $\bar x\geq \underline x$.%
\fnsuchan{ }%
Then
\begin{itemize}
\item $\bar x \geq \widehat x \geq \underline x$ for some $\widehat x\in \argmax_{x\in L} \widetilde G(x,\bar \theta)$, provided the argmax is nonempty, and
\item if $\bar x$ is the largest element of $\argmax_{x\in L} F(x,\bar \theta)$, then $\bar x \geq \widehat x$ for any $\widehat x\in \argmax_{x\in L} \widetilde G(x,\bar \theta)$.
\end{itemize}
\end{namedthm}

The proof is that of \Cref{theorem:lechatelier}, modified along the lines of the proof of \hyperref[theorem:basic-uncertain]{\Cref*{theorem:basic}$'$} (above). We omit the details.

% \begin{proof}[\normalfont\bfseries Proof]
% For the first part, \hyperref[theorem:basic-uncertain]{\Cref*{theorem:basic}}$'$ lets us choose $x' \in \argmax_{x\in L} \widetilde G(x,\bar \theta)$ such that $x'\geq \underline x$. For almost every realization $s$ of $S$, the proof of the first part of \Cref{theorem:lechatelier} yields $F(x',\bar \theta) - C_s(x'-\underline x) \leq F(\bar x \wedge x',\bar \theta) - C_s(\bar x \wedge x'-\underline x)$. Hence $\widetilde G(x',\bar \theta) \leq \widetilde G(\bar x \wedge x',\bar \theta)$ since $u$ is increasing, which since $x'$ maximizes $G(\cdot,\bar \theta)$ on $L$ implies that $\bar x \wedge x'$ does, too. Clearly $\bar x \geq \bar x \wedge x' \geq \underline x$.

% For the second part, fix any $\widehat x\in \argmax_{x\in L} \widetilde G(x,\bar \theta)$. The proof of the second part of \Cref{theorem:lechatelier} yields $\widetilde G(\widehat x,\bar \theta) \geq \widetilde G( \bar x\wedge \widehat x,\bar \theta)$ and, for almost every realization $s$ of $S$, $C_s(\widehat x-\underline x)\geq C_s(\bar x\wedge \widehat x-\underline x)$. It follows that $F(\widehat x,\bar \theta) \geq F(\bar x\wedge \widehat x,\bar \theta)$, whence $F(\bar x\vee \widehat x,\bar \theta) \geq F(\bar x,\bar \theta)$ by quasi-supermodularity. Thus $\bar x \geq \bar x\vee \widehat x$ since $\bar x$ is the largest maximizer of $F(\cdot,\bar \theta)$, which is to say that $\bar x \geq \widehat x$.
% \end{proof}

To extend our dynamic results, augment the dynamic model of \cref{sec:dynamic:setting} as follows. Let the period-$t$ adjustment cost function be $C_{S_t,t}$, where $S_t$ is a random variable. The agent's period-$t$ expected payoff is $\widetilde G_t(x_t,x_{t-1}) = \E\bigl[ u\bigl( F(x_t,\theta_t) - C_{S_t,t}(x_t-x_{t-1}) \bigr) \bigr]$, where $u : \R \to \R$ is increasing. Given her period-$0$ choice $x_0 \in L$, the long-lived agent's problem is to choose a sequence $(x_t)_{t=1}^\infty$ in $L$ to maximize $\widetilde{\mathcal{G}}((x_t)_{t=1}^\infty,x_0) = \sum_{t=1}^{\infty} \delta^{t-1} \widetilde G_t(x_t,x_{t-1})$. We assume that the states $(S_t)_{t=1}^\infty$ are independent across periods.%
\footnote{\hyperref[theorem:lechatelier-dynamic-uncertain]{\Cref*{theorem:lechatelier-dynamic}$'$} is true without independence, but our \emph{interpretation} of it hinges on independence. The result is about sequences $(x_t)_{t=1}^\infty$ that maximize the time-0 expected payoff $\widetilde{\mathcal{G}}(\cdot,x_0)$, not about optimal real-time choice of actions. The time-0 and real-time problems are equivalent if and only if the agent does not learn over time about the realizations of future states, and this is ensured by independence.}

\begin{namedthm}[\Cref*{theorem:lechatelier-dynamic}$\boldsymbol{'}$.]\label{theorem:lechatelier-dynamic-uncertain}
Suppose that the objective $F(x,\theta)$ is quasi-supermodular in $x$ and has single-crossing differences in $(x,\theta)$, and that in each period $t$, at almost every realization $s_t$ of the state $S_t$, the adjustment cost $C_{s_t,t}$ is monotone. Fix $\bar \theta \geq \underline \theta$, and let $\bar x\in \argmax_{x\in L} F(x,\bar \theta)$ satisfy $\bar x\geq \underline x$.%
\fnsuchan{ }%
If $\underline\theta\leq \theta_t\leq \bar\theta$ for every $t \in \N$, then provided the long-lived agent's problem admits a solution, there is a solution $(x_t)_{t=1}^\infty$ that satisfies $\underline x\leq x_t\leq \bar x$ for every period $t \in \N$.
\end{namedthm}

Again, a simple ``a.s.'' modification to the proof \Cref{theorem:lechatelier-dynamic} delivers this result. By contrast, the proof of \Cref{theorem:lechatelier-dynamic-strong} is not easily modified to accommodate uncertain cost, except in the risk-neutral case (when $u$ is affine).

%%%%%%%%%%%%%%%%%%%%%%%%%%%%%%%%%%%
\subsection{Proof of \texorpdfstring{\hyperref[theorem:basic-constraint]{\Cref*{theorem:basic}$^*$}}{Theorem~\ref{theorem:basic}*}}
\label{sec:appendix:pf_basic-constraint}
%%%%%%%%%%%%%%%%%%%%%%%%%%%%%%%%%%%

Let $x'\in \argmax_{x\in \bar L} G(x,\bar \theta)$. Define $\widehat x = \underline x\vee x'$; obviously $\widehat x \geq \underline x$, and $\widehat x$ belongs to $\bar L$ since $\bar L \geq_{\text{ss}} \underline L \ni \underline x$. We claim that $\widehat x$ maximizes $G(\cdot,\bar \theta)$ on $\bar L$. We have $F(\underline x,\underline \theta) \geq F(\underline x \wedge x', \underline \theta)$ by definition of $\underline x$ since $\underline x \wedge x' \in \underline L$ (by $\bar L \geq_{\text{ss}} \underline L \ni \underline x$). Thus $F(\underline x \vee x',\underline \theta) \geq F(x', \underline \theta)$ by quasi-supermodularity, whence $F(\underline x \vee x',\bar \theta) \geq F(x', \bar \theta)$ by single-crossing differences. Furthermore, by minimal monotonicity, $C( \underline x \vee x' - \underline x ) = C( ( x' - \underline x ) \vee 0 ) \leq C( x' - \underline x )$.
Thus
\begin{equation*}
G(\widehat x,\bar \theta)
= F(\underline x\vee x',\bar \theta)-C(\underline x\vee x'-\underline x)
\geq F(x',\bar \theta)-C(x'-\underline x)
= G(x',\bar \theta).
\end{equation*}
Since $x'$ maximizes $G(\cdot,\bar \theta)$ on $\bar L$, it follows that $\widehat x$ does, too. \qed

%%%%%%%%%%%%%%%%%%%%%%%%%%%%%%%%%%%
\subsection{Proof of \texorpdfstring{\Cref{proposition:basic-strict}}{Proposition~\ref{proposition:basic-strict}}}
\label{sec:appendix:pf_basic-strict}
%%%%%%%%%%%%%%%%%%%%%%%%%%%%%%%%%%%

Let $\widehat x \in \argmax_{x\in L} G(x,\bar \theta)$, and suppose toward a contradiction that $\widehat x \ngeq \underline x$. Then $\underline x \vee \widehat x > \widehat x$. The proof of \Cref{theorem:basic} yields $F(\underline x \vee \widehat x,\bar \theta) \geq F(\widehat x, \bar \theta)$ and $C(\underline x\vee \widehat x-\underline x)\leq C(\widehat x-\underline x)$, where the first inequality is strict if the single-crossing differences of $F$ is strict, and the second inequality is strict if the minimal monotonicity of $C$ is strict. In either case, we have
\begin{equation*}
G(\underline x\vee \widehat x,\bar \theta)
= F(\underline x\vee \widehat x,\bar \theta)
- C(\underline x\vee \widehat x-\underline x)
> F(\widehat x,\bar \theta)
- C(\widehat x-\underline x)
= G(\widehat x,\bar \theta) ,
\end{equation*}
which contradicts the fact that $\widehat x$ maximizes $G(\cdot,\bar \theta)$ on $L$. \qed

%%%%%%%%%%%%%%%%%%%%%%%%%%%%%%%%%%%
\subsection{Proof of \texorpdfstring{\Cref{claim:wishful_spm} in \cref{sec:mcs:wishful}}{Claim~\ref{claim:wishful_spm} in section~\ref{sec:mcs:wishful}}}
\label{sec:appendix:pf_wishful_spm}
%%%%%%%%%%%%%%%%%%%%%%%%%%%%%%%%%%%

Label the elements of $\mathcal{Y}$ as $\mathcal{Y} = \{y_1,\dots,y_N\}$, where $y_1 < \cdots < y_N$. Let $\phi(c,y) = u_1(c) + u_2\left( (1+r) (w-c) + y \right)$ for each $c \in [0,w]$ and $y \in \mathcal{Y}$. The function $\phi$ is supermodular since $u_2$ is concave. For all $c \in [0,w]$ and $G \in \mathcal{G}$, 
\begin{equation*}
    U(c,G)
    = \int_{\mathcal{Y}} \phi(c,y) G(\mathrm{d} y)
    = \phi(c,y_N) + \sum_{n=1}^{N-1} \bigl[ \phi(c,y_{n+1}) - \phi(c,y_n) \bigr] \bigl[ -G(y_n) \bigr]
\end{equation*}
from telescoping / integration by parts. Thus for any $c,c' \in [0,w]$ and $G,H \in \mathcal{G}$, letting $k = \phi(c,y_N) - \phi(c \wedge c',y_N) = \phi(c \vee c',y_N) - \phi(c',y_N)$, we have
\begin{multline*}
U(c,G) - U(c \wedge c', G \wedge_1 H)
\\
\begin{aligned}
&= k + \sum_{n=1}^{N-1}
\bigl[ \phi(c,y_{n+1}) - \phi(c \wedge c',y_n) \bigr]
\bigl[ (G \wedge_1 H)(y_n) - G(y_n) \bigr]
\\
&= k + \sum_{n=1}^{N-1}
\bigl[ \phi(c,y_{n+1}) - \phi(c \wedge c',y_n) \bigr]
\bigl[ H(y_n) - (G \vee_1 H)(y_n) \bigr]
\\
&\leq k + \sum_{n=1}^{N-1}
\bigl[ \phi(c \vee c',y_{n+1}) - \phi(c',y_n) \bigr]
\bigl[ H(y_n) - (G \vee_1 H)(y_n) \bigr]
\\
&= U(c \vee c', G \vee_1 H) - U(c',H) ,
\end{aligned}
\end{multline*}
where the inequality holds since $\phi$ is supermodular and $H \leq_1 G \vee_1 H$. \qed

%%%%%%%%%%%%%%%%%%%%%%%%%%%%%%%%%%%
\subsection{Proof of \texorpdfstring{\Cref{proposition:lechatelier-strict}}{Proposition~\ref{proposition:lechatelier-strict}}}
\label{sec:appendix:pf_lechatelier-strict}
%%%%%%%%%%%%%%%%%%%%%%%%%%%%%%%%%%%

Fix any $\widehat x \in \argmax_{x\in L} G(x,\bar \theta)$. We have $\widehat x \geq \underline x$ by \Cref{proposition:basic-strict}. To show that $\bar x \geq \widehat x$, suppose toward a contradiction that $\bar x \ngeq \widehat x$. The proof of \Cref{theorem:lechatelier} yields $F(\widehat x,\bar \theta) \leq F(\bar x\wedge \widehat x,\bar \theta)$ and $C(\widehat x-\underline x) < C(\bar x\wedge \widehat x-\underline x)$, where the latter inequality is strict since $\widehat x-\underline x \neq \bar x\wedge \widehat x-\underline x$ and $C$ is \emph{strictly} monotone. Thus
\begin{equation*}
G(\bar x\wedge \widehat x,\bar \theta)
= F(\bar x\wedge \widehat x,\bar \theta)
- C(\bar x\wedge \widehat x-\underline x)
> F(\widehat x,\bar \theta)
- C(\widehat x-\underline x)
= G(\widehat x,\bar \theta) ,
\end{equation*}
which contradicts the fact that $\widehat x$ maximizes $G(\cdot,\bar \theta)$ on $L$. \qed

%%%%%%%%%%%%%%%%%%%%%%%%%%%%%%%%%%%
\subsection{A monotonicity lemma}
\label{sec:appendix:mon_impl}
%%%%%%%%%%%%%%%%%%%%%%%%%%%%%%%%%%%

The following lemma will be used below in the proofs of \Cref{theorem:lechatelier-dynamic,theorem:myopic-vs-fwd}.

\begin{lemma}\label{lemma:mon_impl2}
If $C : \Delta L \to [0,\infty]$ is monotone, then $C( z \vee x - z \vee y ) \leq C( x - y ) \geq C( z \wedge x - z \wedge y )$ for any $x,y,z \in L$.
\end{lemma}

\begin{proof}[\normalfont\bfseries Proof]
We shall prove the first inequality; the second follows similarly. By monotonicity, it suffices to show that for each $i$, one of the following holds:
\begin{enumerate}[label=(\alph*)]
\item \label{bullet:mon_impl2_a} $0 \leq ( z \vee x - z \vee y )_i \leq ( x - y )_i$
\item \label{bullet:mon_impl2_b} $0 \geq ( z \vee x - z \vee y )_i \geq ( x - y )_i$.
\end{enumerate}
If $x_i \geq z_i \geq y_i$ then \ref{bullet:mon_impl2_a} holds by inspection, if $x_i \leq z_i \leq y_i$ then \ref{bullet:mon_impl2_b} holds by inspection, if $x_i \leq z_i \geq y_i$ then \ref{bullet:mon_impl2_a} or \ref{bullet:mon_impl2_b} holds since $( z \vee x - z \vee y )_i = 0$, and if $x_i \geq z_i \leq y_i$ then \ref{bullet:mon_impl2_a} or \ref{bullet:mon_impl2_b} holds since $( z \vee x - z \vee y )_i = ( x - y )_i$.
\end{proof}

The following corollary will be used below to prove \Cref{theorem:lechatelier-dynamic-myopic,theorem:myopic-vs-fwd}.

\begin{corollary}\label{corollary:mon_impl}
Let $C : \Delta L \to [0,\infty]$ be monotone, and consider $x,y,z \in L$.
\begin{itemize}
\item If $z \leq y$, then $C( z \vee x - y ) \leq C( x - y )$.
\item If $z \geq y$, then $C( z \wedge x - y ) \leq C( x - y )$.
\end{itemize}
\end{corollary}

%%%%%%%%%%%%%%%%%%%%%%%%%%%%%%%%%%%
\subsection{Proof of \texorpdfstring{\Cref{theorem:lechatelier-dynamic}}{Theorem~\ref{theorem:lechatelier-dynamic}}}
\label{sec:appendix:pf_lechatelier-dynamic}
%%%%%%%%%%%%%%%%%%%%%%%%%%%%%%%%%%%

Let $(x_t)_{t=1}^\infty$ maximize $\mathcal{G}(\cdot,x_0)$. We shall show that $( \bar x \wedge (\underline x \vee x_t) )_{t=1}^\infty$ also maximizes $\mathcal{G}(\cdot,x_0)$; this suffices since $\underline x \leq \bar x \wedge (\underline x \vee x_t) \leq \bar x$ for each $t$.

We first show that $(\widehat x_t)_{t=1}^\infty = (\underline x \vee x_t)_{t=1}^\infty$ maximizes $\mathcal{G}(\cdot,x_0)$. For every $t \in \N$, we have $F(\underline x,\underline \theta)\geq F(\underline x\wedge x_t,\underline\theta)$ by definition of $\underline x$, which by quasi-supermodularity implies that $F(\underline x\vee x_t,\underline \theta)\geq F(x_t,\underline\theta)$, whence $F(\underline x\vee x_t,\theta_t)\geq F(x_t,\theta_t)$ by single-crossing differences and $\theta_t \geq \underline \theta$. Thus $\mathcal{F}((\underline x \vee x_t)_{t=1}^\infty)\geq \mathcal{F}((x_t)_{t=1}^\infty)$. Furthermore, for every $t \in \N$, we have $C_t( \underline x \vee x_t - \underline x \vee x_{t-1} ) \leq C_t( x_t - x_{t-1} )$ by \Cref{lemma:mon_impl2} (\cref{sec:appendix:mon_impl}) since $C_t$ is monotone, so $\mathcal{C}(x_0,(\underline x \vee x_t)_{t=1}^\infty)\leq \mathcal{C}(x_0,(x_t)_{t=1}^\infty)$. So $\mathcal{G}((\underline x \vee x_t)_{t=1}^\infty,x_0)\geq \mathcal{G}((x_t)_{t=1}^\infty,x_0)$, which since $(x_t)_{t=1}^\infty$ maximizes $\mathcal{G}(\cdot,x_0)$ implies that $(\widehat x_t)_{t=1}^\infty = (\underline x \vee x_t)_{t=1}^\infty$ does, too.

It remains to show that $(\bar x\wedge \widehat x_t)_{t=1}^\infty$ also maximizes $\mathcal{G}(\cdot,x_0)$. For every $t \in \N$, we have $F(\bar x\vee \widehat x_t,\bar\theta) \leq F(\bar x,\bar \theta)$ by definition of $\bar x$, which by quasi-supermodularity implies that $F(\widehat x_t,\bar\theta) \leq F(\bar x\wedge \widehat x_t,\bar \theta)$, whence $F(\widehat x_t,\theta_t) \leq F(\bar x\wedge \widehat x_t,\theta_t)$ by single-crossing differences and $\theta_t \leq \bar \theta$. Thus $\mathcal{F}((\widehat x_t)_{t=1}^\infty) \leq \mathcal{F}((\bar x \wedge \widehat x_t)_{t=1}^\infty)$. Furthermore, for every $t \in \N$, we have $C_t( \widehat x_t - \widehat x_{t-1} ) \geq C_t( \bar x \wedge \widehat x_t - \bar x \wedge \widehat x_{t-1} )$ by \Cref{lemma:mon_impl2} since $C_t$ is monotone, so $\mathcal{C}(x_0,(\widehat x_t)_{t=1}^\infty) \geq \mathcal{C}(x_0,(\bar x \wedge \widehat x_t)_{t=1}^\infty)$. So $\mathcal{G}((\widehat x_t)_{t=1}^\infty,x_0) \leq \mathcal{G}((\bar x \wedge \widehat x_t)_{t=1}^\infty,x_0)$, which since $(\widehat x_t)_{t=1}^\infty$ maximizes $\mathcal{G}(\cdot,x_0)$ implies that $(\bar x \wedge \widehat x_t)_{t=1}^\infty$ does, too. \qed

%%%%%%%%%%%%%%%%%%%%%%%%%%%%%%%%%%%
\subsection{Proof of \texorpdfstring{\Cref{theorem:lechatelier-dynamic-strong}}{Theorem~\ref{theorem:lechatelier-dynamic-strong}}}
\label{sec:appendix:pf_lechatelier-dynamic-strong}
%%%%%%%%%%%%%%%%%%%%%%%%%%%%%%%%%%%

For any sequence $\boldsymbol{x} = (x_t)_{t=1}^\infty$ in $L$ and any $T \in \N$, let $M_T \boldsymbol{x}$ denote the sequence in $L$ whose $t^\text{th}$ entry is $x_t$ for $t<T$ and $x_{t-1} \vee x_t$ for $t \geq T$.

Assume that the agent's problem admits a solution. Let $\boldsymbol{x}^1 = (x^1_t)_{t=1}^\infty$ be a solution satisfying $\underline x\leq x^1_t\leq \bar x$ in every period $t$; such a solution exists by \Cref{theorem:lechatelier-dynamic}. Define $X_t = x^1_1 \vee x^1_2 \vee \cdots \vee x^1_{t-1} \vee x^1_t$ for $t \in \N$, and $X_0 = \underline x$.

Write $\boldsymbol{x}^T = M_T M_{T-1} \cdots M_3 M_2 \boldsymbol{x}^1$ for $T \geq 2$. By inspection, the first $T$ entries of $\boldsymbol{x}^T$ are $X_1,X_2,\dots,X_{T-1},X_T$.
Clearly $\underline x \leq X_t \leq X_{t+1} \leq \bar x$ for any period $t \in \N$. To prove the theorem, we need only show that $\boldsymbol{x}^\infty = (X_1,X_2,X_3,\dots)$ is optimal.

It suffices to show for each $T \in \N$ that $\boldsymbol{x}^T$ is optimal. For then, letting $V$ be the optimal value and noting that both $\boldsymbol{x}^T = (x_t)_{t=1}^\infty$ and $\boldsymbol{x}^\infty$ have $X_1,\dots,X_T$ as their first $T$ entries, we have
\begin{align}
0
\geq \mathcal{G}(\boldsymbol{x}^\infty,x_0)
- V
&= \mathcal{G}(\boldsymbol{x}^\infty,x_0)
- \mathcal{G}(\boldsymbol{x}^T,x_0)
\nonumber\\
&= \delta^T \left[
\mathcal{G}((X_t)_{t=T+1}^\infty,X_T)
- \mathcal{G}((x_t)_{t=T+1}^\infty,X_T)
\right]
\nonumber\\
&= \delta^T \left[
\mathcal{F}((X_t)_{t=T+1}^\infty)
- \mathcal{F}((x_t)_{t=T+1}^\infty)
\right]
\nonumber\\
&\qquad- \delta^T \left[
\mathcal{C}(X_T,(X_t)_{t=T+1}^\infty)
- \mathcal{C}(X_T,(x_t)_{t=T+1}^\infty)
\right]
\nonumber\\
&\geq \delta^T \left[
\mathcal{F}((X_t)_{t=T+1}^\infty)
- \mathcal{F}((x_t)_{t=T+1}^\infty)
\right] ,
\label{eq:fwd_bound}
\end{align}
where the final inequality holds since
\begin{align*}
C( X_t - X_{t-1} ) &\leq C( x_t - x_{t-1} )
\quad \text{for every $t \geq T+2$}
\\
\text{and}\quad
C( X_{T+1} - X_T ) &\leq C( x_{T+1} - X_T )
\end{align*}
by the monotonicity of $C$. (For $t \geq T+2$, for each dimension $i$, if $X_{t,i} = X_{t-1,i}$ then $0 = ( X_t - X_{t-1} )_i$, while if $X_{t,i} > X_{t-1,i}$ then $x_{t,i} = X_{t,i} \geq X_{t-1,i} \geq x_{t-1,i}$, so $0 \leq ( X_t - X_{t-1} )_i \leq ( x_t - x_{t-1} )_i$. For the $t=T+1$ inequality, if $X_{T+1,i} = X_{T,i}$ then $0 = ( X_{T+1} - X_T )_i \geq ( x_{T+1} - X_T )_i$, while if $X_{T+1,i} > X_{T,i}$ then $X_{T+1,i} = x_{T+1,i}$, so $( X_{T+1} - X_T )_i = ( x_{T+1} - X_T )_i$.) Since $F(\cdot,\bar \theta)$ is BCS, the ``${[\cdot]}$'' expression in \eqref{eq:fwd_bound} is bounded below uniformly over $T \in \N$,%
\footnote{Since $X_t$ and $x_t$ belong to the compact set $[ \underline x, \bar x]$ for every $t \in \N$, there is a $K>0$ such that $F(X_t,\bar \theta) - F(x_t,\bar \theta) \geq -2K$ for all $t$, so ``$[\cdot]$'' is bounded below by $-2K / (1-\delta)$.}
so letting $T \to \infty$ yields $0 \geq \mathcal{G}(\boldsymbol{x}^\infty,x_0) - V \geq 0$, which is to say that $\boldsymbol{x}^\infty$ is optimal.

To show that $\boldsymbol{x}^T$ is optimal for each $T \in \N$, we employ induction on $T \in \N$. The base case $T=1$ is immediate.

For the induction step, fix any $T \in \N$, and suppose that $\boldsymbol{x}^T = (x_t)_{t=1}^\infty$ is optimal; we will show that $\boldsymbol{x}^{T+1} = M_{T+1} \boldsymbol{x}^T$ is also optimal. Let $(\widetilde{x}_t)_{t=1}^\infty$ be the sequence with $t^\text{th}$ entry $x_t$ for $t<T$ and $x_t \wedge x_{t+1}$ for $t \geq T$. Since $\boldsymbol{x}^T = (x_t)_{t=1}^\infty$ is optimal, and $(\widetilde{x}_t)_{t=1}^\infty$ shares its first $T-1$ entries $X_1,\dots,X_{T-1}$, we have $\mathcal{G}((x_t)_{t=T}^\infty,X_{T-1})
\geq \mathcal{G}(({\widetilde x}_t)_{t=T}^\infty,X_{T-1})$, which may be written in full as
\begin{multline}\label{start}
\textstyle\sum_{t=T}^\infty \delta^{t-T} \Bigl( \left[F(x_t,\bar \theta)-F(x_t\wedge x_{t+1},\bar \theta)\right]\\
-\left[C(x_t-x_{t-1})-C(x_t\wedge x_{t+1}-x_{t-1}\wedge x_t)\right] \Bigr)\geq 0.
\end{multline}
(Note that since $x_t = X_t$ for every $t \leq T$, we have $x_{T-1} \wedge x_T = X_{T-1} = x_{T-1}$.) Since $F(\cdot, \bar \theta)$ is supermodular, it holds for every $t \geq T$ that
\begin{equation}\label{start1}
F(x_t\vee x_{t+1},\bar \theta)-F(x_{t+1},\bar \theta)
\geq F(x_t,\bar \theta)-F(x_t\wedge x_{t+1},\bar \theta)
\end{equation}
We furthermore claim that for each $t \geq T$,
\begin{multline} \label{start2}
C(x_t\vee x_{t+1}-x_{t-1}\vee x_t)-C(x_{t+1}-x_t)
\\
\leq C(x_t-x_{t-1})-C(x_t\wedge x_{t+1}-x_{t-1}\wedge x_t) ;
\end{multline}
we shall prove this shortly. Combining \eqref{start}, \eqref{start1} and \eqref{start2}, and changing variables in the sum, we obtain
\begin{multline*}
\textstyle\sum_{t=T+1}^{\infty} \delta^{t-(T+1)} \Bigl( \left[F(x_{t-1}\vee x_t,\bar \theta)-F(x_t,\bar \theta)\right]\\
-\left[C(x_{t-1}\vee x_t-x_{t-2}\vee x_{t-1})-C(x_t-x_{t-1})\right] \Bigr)
\geq 0 .
\end{multline*}
By inspection, this says precisely that $(\widehat x_t)_{t=1}^\infty = \boldsymbol{x}^{T+1} = M_{T+1} \boldsymbol{x}^T$ satisfies
\begin{equation*}
\mathcal{G}((\widehat x_t)_{t=T+1}^\infty,X_T)
\geq \mathcal{G}((x_t)_{t=T+1}^\infty,X_T).
\end{equation*}
(Note that since $x_t = X_t$ for every $t \leq T$, we have $x_{(T+1)-2} \vee x_{(T+1)-1} = X_{T-1} \vee X_T = X_T = x_{(T+1)-1}$.) Since $\boldsymbol{x}^{T+1} = (\widehat x_t)_{t=1}^\infty$ and $\boldsymbol{x}^T = (x_t)_{t=1}^\infty$ agree in their first $T$ entries, and $\boldsymbol{x}^T$ is optimal, it follows that $\boldsymbol{x}^{T+1}$ is optimal, too.

It remains to show that \eqref{start2} holds. It suffices to prove for each $i$ that
\begin{equation}\label{delicate}
C_i(y\vee z-x\vee y)+C_i(y\wedge z-x\wedge y)\leq C_i(y-x)+C_i(z-y)
\quad \text{for any $x,y,z$.}
\end{equation}
(We've renamed $x_{t-1,i}=x$, $x_{t,i}=y$ and $x_{t+1,i}=z$.) When $y$ is not extreme (neither least nor greatest), \eqref{delicate} holds trivially because the left-hand side is equal to the right-hand side. When $y$ is extreme, \eqref{delicate} reads
\begin{equation*}
C_i(0) + C_i(z-x) \leq C_i(y-x) + C_i(z-y) ,
\end{equation*}
and we have either
\begin{align*}
&&\text{(i)}\quad &0 \leq z-x \leq y-x
\qquad \text{or} \quad
&\text{(ii)}\quad &0 \leq z-x \leq z-y
\\
\text{or} \quad
&&\text{(iii)}\quad &0 \geq z-x \geq y-x
\qquad \text{or} \quad
&\text{(iv)}\quad &0 \geq z-x \geq z-y .
\end{align*}
Since $C_i$ is single-dipped and minimized at zero, we have $C_i(z-x) \leq C_i(y-x) \leq C_i(y-x) + C_i(z-y) - C_i(0)$ in the first and third cases, and $C_i(z-x) \leq C_i(z-y) \leq C_i(y-x) + C_i(z-y) - C_i(0)$ in the second and fourth. \qed

\vspace{.5\baselineskip}

The above proof applies nearly unchanged if there is a finite horizon $K$, so that the agent chooses a length-$K$ sequence $(x_t)_{t=1}^K$ in $L$ to maximize $\sum_{t=1}^K \delta^{t-1} \left[ F(x_t,\bar \theta) - C(x_t-x_{t-1}) \right]$. The limit argument early on is superfluous, so BCS is not needed. Let $x_{K+1} = x_K$ and follow the same steps to obtain
\begin{multline*}
\textstyle\sum_{t=T}^K \delta^{t-T} \Bigl( \left[F(x_t,\bar \theta)-F(x_t\wedge x_{t+1},\bar \theta)\right]\\
- \left[C(x_t-x_{t-1})-C(x_t\wedge x_{t+1}-x_{t-1}\wedge x_t)\right] \Bigr) \geq 0.
\end{multline*}
This together with \eqref{start1} and \eqref{start2} and a change of variable delivers
\begin{multline*}
\textstyle\sum_{t=T+1}^{K+1} \delta^{t-(T+1)} \Bigl( \left[F(x_{t-1}\vee x_t,\bar \theta)-F(x_t,\bar \theta)\right]\\
- \left[C(x_{t-1}\vee x_t-x_{t-2}\vee x_{t-1})-C(x_t-x_{t-1})\right] \Bigr)
\geq 0 .
\end{multline*}
The sum's final term equals $\delta^{K-(T+1)} \bigl( 0 - \left[ C(x_K-x_{K-1}\vee x_K)-C(0) \right] \bigr)$, which is nonpositive since $C$ is minimized at $0$ by monotonicity. Hence the inequality is preserved when the final term is dropped:
\begin{multline*}
\textstyle\sum_{t=T+1}^K \delta^{t-(T+1)} \Bigl( \left[F(x_{t-1}\vee x_t,\bar \theta)-F(x_t,\bar \theta)\right]\\
- \left[C(x_{t-1}\vee x_t-x_{t-2}\vee x_{t-1})-C(x_t-x_{t-1})\right] \Bigr)
\geq 0 .
\end{multline*}
The remainder of the argument now applies unchanged.

%%%%%%%%%%%%%%%%%%%%%%%%%%%%%%%%%%%
\subsection{Proof of \texorpdfstring{\Cref{theorem:lechatelier-dynamic-myopic}}{Theorem~\ref{theorem:lechatelier-dynamic-myopic}}}
\label{sec:appendix:pf_lechatelier-dynamic-myopic}
%%%%%%%%%%%%%%%%%%%%%%%%%%%%%%%%%%%
\renewcommand\qedsymbol{\itshape QED}

For the first part, fix a sequence $(\theta_t)_{t=1}^\infty$ in $\Theta$ such that $\underline \theta \leq \theta_t \leq \bar \theta$ for every $t \in \N$. Call a finite sequence $(x_t)_{t=1}^T$ \emph{equilibrium caged} if and only if $x_t \in \argmax_{x \in L} G_t(x,x_{t-1})$ and $\underline x \leq x_t \leq \bar x$ for every $t \in \{1,\dots,T\}$. By \Cref{theorem:lechatelier}, there exists an equilibrium caged sequence of length $T=1$. Given this, it suffices to prove that for each $T \geq 2$, any length-$(T-1)$ equilibrium caged sequence $(x_t)_{t=1}^{T-1}$ may be extended to a length-$T$ equilibrium caged sequence $(x_t)_{t=1}^T$ (by appropriately choosing $x_T$). To that end, fix an arbitrary $T \geq 2$, and let $(x_t)_{t=1}^{T-1}$ be equilibrium caged. We shall prove two claims:

\begin{claim} \label{claim:lc-myopic_1}
There is an $x'\in \argmax_{x\in L} G_T(x,x_{T-1})$ such that $x' \geq \underline x$.
\end{claim}

\begin{claim} \label{claim:lc-myopic_2}
$\bar x\wedge x'$ belongs to $\argmax_{x\in L} G_T(x,x_{T-1})$ whenever $x'$ does.
\end{claim}

These claims suffice because $x_T = \bar x\wedge x'$ satisfies $\underline x \leq x_T \leq \bar x$.

\begin{proof}[Proof of \Cref{claim:lc-myopic_1}]
% \footnote{This argument is close to the proof of \Cref{theorem:basic}.}}
Fix any $x''\in \argmax_{x\in L} G_T(x,x_{T-1})$. We will show that $x' = \underline x\vee x''$ also maximizes $G_T(\cdot,x_{T-1})$; obviously $x' \geq \underline x$. We have $F(\underline x,\underline \theta) \geq F(\underline x \wedge x'', \underline \theta)$ by definition of $\underline x$. Thus $F(\underline x \vee x'',\underline \theta) \geq F(x'', \underline \theta)$ by quasi-super\-modu\-larity, whence $F(\underline x \vee x'',\theta_T) \geq F(x'', \theta_T)$ by single-crossing differences and $\theta_T \geq \underline \theta$. Furthermore, since $C_T$ is monotone and $\underline x \leq x_{T-1}$, we have $C_T( \underline x \vee x' - x_{T-1} ) \leq C_T( x' - x_{T-1} )$ by \Cref{corollary:mon_impl} (\cref{sec:appendix:mon_impl}). Thus
\begin{align*}
G_T(x',x_{T-1})
&= F(\underline x\vee x'',\theta_T)-C_T(\underline x\vee x''-x_{T-1})
\\
&\geq F(x'',\theta_T)-C_T(x''-x_{T-1})
= G_T(x'',x_{T-1}) .
\end{align*}
Since $x''$ maximizes $G_T(\cdot,x_{T-1})$ on $L$, it follows that $x'$ does, too.
\end{proof}

\begin{proof}[Proof of \Cref{claim:lc-myopic_2}]
% \footnote{This argument is close to the proof of \Cref{theorem:lechatelier}.}}
Let $x'$ belong to $\argmax_{x\in L} G_T(x,x_{T-1})$; we claim that $\widehat x = \bar x\wedge x'$ also maximizes $G_T(\cdot,x_{T-1})$. We have $F(\bar x\vee x',\bar \theta) \leq F(\bar x,\bar \theta)$ by definition of $\bar x$, whence $F(x',\bar \theta) \leq F(\bar x\wedge x',\bar \theta)$ by quasi-supermodularity, so that $F(x',\theta_T) \leq F(\bar x\wedge x',\theta_T)$ by single-crossing differences and $\theta_T \leq \bar \theta$. Since $C_T$ is monotone and $\bar x \geq x_{T-1}$, we have $C_T(x'-x_{T-1})\geq C_T(\bar x\wedge x'-x_{T-1})$ by \Cref{corollary:mon_impl}. Thus
\begin{align*}
G_T(x',x_{T-1})
&= F(x',\theta_T)-C_T(x'-x_{T-1})
\\
&\leq F(\bar x\wedge x',\theta_T)-C_T(\bar x\wedge x'-x_{T-1})
= G_T(\widehat x,x_{T-1}) ,
\end{align*}
which since $x'$ maximizes $G_T(\cdot,x_{T-1})$ on $L$ implies that $\widehat x$ does, too.
\end{proof}

To prove the second part of \Cref{theorem:lechatelier-dynamic-myopic}, fix a sequence $(\theta_t)_{t=1}^\infty$ in $\Theta$ such that $\underline \theta \leq \theta_t \leq \theta_{t+1} \leq \bar \theta$ for every $t \in \N$. Recall that $x_0 = \underline x$. Call a finite sequence $(x_t)_{t=1}^T$ \emph{equilibrium monotone} if and only if $x_t \in \argmax_{x \in L} G_t(x,x_{t-1})$ and $x_{t-1} \leq x_t \leq \bar x$ for every $t \in \{1,\dots,T\}$. By \Cref{theorem:lechatelier}, there exists an equilibrium monotone sequence of length $T=1$. Given this, it suffices to prove that for every $T \geq 2$, any length-$(T-1)$ equilibrium monotone sequence $(x_t)_{t=1}^{T-1}$ may be extended to a length-$T$ equilibrium monotone sequence $(x_t)_{t=1}^T$ (by an appropriate choice of $x_T$).

To that end, fix an arbitrary $T \geq 2$, and let $(x_t)_{t=1}^{T-1}$ be equilibrium monotone; we shall show that for every $t \in \{1,\dots,T\}$, there is an $x' \in L$ which belongs to $\argmax_{x \in L} G_T(x,x_{T-1})$ and satisfies $x_{t-1} \leq x' \leq \bar x$. We proceed by induction on $t \in \{1,\dots,T\}$. The base case $t=1$ follows from \Cref{claim:lc-myopic_1,claim:lc-myopic_2}. For the induction step, suppose that there is an $x'' \in \argmax_{x\in L} G_T(x,x_{T-1})$ that satisfies $x_{t-2} \leq x'' \leq \bar x$; we claim that $x' = x_{t-1}\vee x''$ also maximizes $G_T(\cdot,x_{T-1})$. This suffices since $x_{t-1} \leq x' \leq \bar x$, where the latter inequality holds because $x_{t-1} \leq \bar x$ (as $(x_s)_{s=1}^{T-1}$ is equilibrium monotone) and $x'' \leq \bar x$.

We have $G_{t-1}(x_{t-1},x_{t-2}) \geq G_{t-1}(x_{t-1}\wedge x'',x_{t-2})$ by definition of $x_{t-1}$. Since $C_{t-1}$ is monotone and $x'' \geq x_{t-2}$ by the induction hypothesis, we have $C_{t-1}(x_{t-1}-x_{t-2})\geq C_{t-1}(x_{t-1}\wedge x''-x_{t-2})$ by \Cref{corollary:mon_impl} (\cref{sec:appendix:mon_impl}). It follows that $F(x_{t-1},\theta_{t-1}) \geq F(x_{t-1}\wedge x'',\theta_{t-1})$. Thus $F(x_{t-1}\vee x'',\theta_{t-1}) \geq F(x'',\theta_{t-1})$ by quasi-supermodularity, whence $F(x_{t-1}\vee x'',\theta_T) \geq F(x'',\theta_T)$ by single-crossing differences and $\theta_T \geq \theta_{t-1}$. We have $C_T( x_{t-1}\vee x'' - x_{T-1} ) \leq C_T(x'' - x_{T-1})$ by \Cref{corollary:mon_impl} since $C_T$ is monotone and $x_{t-1} \leq x_{T-1}$, where the latter holds since $(x_s)_{s=1}^{T-1}$ is equilibrium monotone. Thus
\begin{align*}
G_T(x',x_{T-1})
&= F(x_{t-1}\vee x'',\theta_T)-C_T(x_{t-1}\vee x''-x_{T-1})
\\
&\geq F(x'',\theta_T)-C_T(x''-x_{T-1})
= G_T(x'',x_{T-1}) ,
\end{align*}
which since $x''$ maximizes $G_T(\cdot,x_{T-1})$ on $L$ implies that $x'$ does, too. \renewcommand\qedsymbol{\bfseries QED} \qed

%%%%%%%%%%%%%%%%%%%%%%%%%%%%%%%%%%%
\subsection{Proof of \texorpdfstring{\Cref{proposition:lechatelier-medium}}{Proposition~\ref{proposition:lechatelier-medium}}}
\label{sec:appendix:pf_lechatelier-medium}
%%%%%%%%%%%%%%%%%%%%%%%%%%%%%%%%%%%

Define a sequence $(\theta_t)_{t=1}^\infty$ in $\Theta$ by $\theta_t = \bar \theta$ for every $t$. For each $t \geq 3$, define $C_t : \Delta L \to [0,\infty]$ by $C_t(\varepsilon) = \infty$ for every $\varepsilon \neq 0$ and $C_t(0)=0$. Now apply \Cref{theorem:lechatelier-dynamic-myopic}. \qed

%%%%%%%%%%%%%%%%%%%%%%%%%%%%%%%%%%%
\subsection{Proof of \texorpdfstring{\Cref{theorem:myopic-vs-fwd}}{Theorem~\ref{theorem:myopic-vs-fwd}}}
\label{sec:appendix:pf_lechatelier-myopic-vs-fwd}
%%%%%%%%%%%%%%%%%%%%%%%%%%%%%%%%%%%

For any sequence $\boldsymbol{x} = (x_t)_{t=1}^\infty$ in $L$ and any $T \in \{0,1,2,\dots\}$, let $R_T \boldsymbol{x}$ denote the sequence in $L$ whose $t^\text{th}$ entry is $x_t$ for $t<T$ and $\widetilde x_T \vee x_t$ for $t \geq T$.

The long-lived agent's problem is to maximize $\mathcal{G}(\cdot,x_0)$. Assume that it admits a solution. Let $\boldsymbol{x'} = (x'_t)_{t=1}^\infty$ be a solution satisfying $\underline x\leq x'_t\leq \bar x$ in every period $t$; such a solution exists by \Cref{theorem:lechatelier-dynamic}. Define $X_t = \widetilde x_t \vee x'_t$ for each $t \in \N$, and $X_0 = \underline x$.

Write $\boldsymbol{x}^T = R_T R_{T-1} \cdots R_2 R_1 R_0 \boldsymbol{x'}$ for $T \in \{0,1,2,\dots\}$. The sequence $\boldsymbol{x}^T$ has $t^\text{th}$ entry $X_t$ for $t \leq T$ and $\widetilde x_T \vee x_t$ for $t > T$, since $(\widetilde x_t)_{t=1}^\infty$ is increasing. Clearly $\widetilde x_t \leq X_t \leq \bar x$ for any period $t \in \N$. To prove the theorem, we need only show that $\boldsymbol{x}^\infty = (X_1,X_2,X_3,\dots)$ maximizes $\mathcal{G}(\cdot,x_0)$.

It suffices to show for each $T \in \{0,1,2,\dots\}$ that $\boldsymbol{x}^T$ maximizes $\mathcal{G}(\cdot,x_0)$. For then, letting $V$ be the long-lived agent's optimal value and noting that both $\boldsymbol{x}^T = (x_t)_{t=1}^\infty$ and $\boldsymbol{x}^\infty$ have $X_1,\dots,X_T$ as their first $T$ entries, we have
\begin{align*}
0
\geq \mathcal{G}(\boldsymbol{x}^\infty,x_0)
- V
&= \mathcal{G}(\boldsymbol{x}^\infty,x_0)
- \mathcal{G}(\boldsymbol{x}^T,x_0)
\\
&= \delta^T \left[
\mathcal{G}((X_t)_{t=T+1}^\infty,X_T)
- \mathcal{G}((x_t)_{t=T+1}^\infty,X_T)
\right] .
\end{align*}
By equi-BCS, the right-hand ``${[\cdot]}$'' is bounded below uniformly over $T \in \N$,%
\footnote{Since $X_t$ and $x_t$ belong to the compact set $[ \underline x, \bar x]$ for every $t \in \N$, there are constants $A,B>0$ such that $F(X_t,\theta_t) - F(x_t,\theta_t) \geq -2A$ and $- [ C(X_t-X_{t-1}) + C(x_t-x_{t-1}) ] \geq -2B$ for all $t \in \N$, so the right-hand ``$[\cdot]$'' is bounded below by $-2(A+B) / (1-\delta)$.}
so letting $T \to \infty$ yields $0 \geq \mathcal{G}(\boldsymbol{x}^\infty,x_0) - V \geq 0$, meaning that $\boldsymbol{x}^\infty$ is optimal.

To show that $\boldsymbol{x}^T$ is optimal for each $T \in \{0,1,2,\dots\}$, we employ induction on $T \in \{0,1,2,\dots\}$. The base case $T=0$ is immediate, since $\boldsymbol{x}^0 = R_0 \boldsymbol{x'} = \boldsymbol{x'}$ is optimal.

For the induction step, fix any $T \in \N$, and suppose that $\boldsymbol{x}^{T-1} = (x_t)_{t=1}^\infty$ is optimal; we will show that $\boldsymbol{x}^T = R_T \boldsymbol{x}^{T-1}$ is also optimal. Since $\boldsymbol{x}^{T-1}$ and $\boldsymbol{x}^T$ have the same first $T-1$ entries (namely, $X_1,X_2,\dots,X_{T-1}$), and since for $t \geq T$ the $t^\text{th}$ entry of $\boldsymbol{x}^T$ is $\widetilde x_T \vee x_t$, it suffices to show that
\begin{align}
G_T( X_T, X_{T-1} ) &\geq G_T( x_T, X_{T-1} )
\label{eq:myopic-vs-fwd_induction_T}
\\
G_t( \widetilde x_T \vee x_t, \widetilde x_T \vee x_{t-1} ) &\geq G_t( x_t, x_{t-1} )
\qquad\text{for all $t \geq T+1$.}
\label{eq:myopic-vs-fwd_induction_t}
\end{align}

For \eqref{eq:myopic-vs-fwd_induction_T}, since $C_T$ is convex and $X_T \geq x_T$ and $X_{T-1} \geq \widetilde x_{T-1}$, we have
\begin{equation*}
C_T( X_T - X_{T-1} )
- C_T( x_T - X_{T-1} )
\leq C_T( X_T - \widetilde x_{T-1} )
- C_T( x_T - \widetilde x_{T-1} ) .
\end{equation*}
It follows that
\begin{align*}
G_T( X_T, X_{T-1} )
- G_T( x_T, X_{T-1} )
&\geq G_T( X_T, \widetilde x_{T-1} )
- G_T( x_T, \widetilde x_{T-1} )
\\
&\geq
G_T( \widetilde x_T, \widetilde x_{T-1} )
- G_T( \widetilde x_T \wedge x_T, \widetilde x_{T-1} )
\geq 0 ,
\end{align*}
where the second inequality holds since
\begin{align*}
F( X_T, \theta_T ) - F( x_T, \theta_T )
&\geq F( \widetilde x_T, \theta_T ) - F( \widetilde x_T \wedge x_T, \theta_T ) \qquad \text{and}
\\
C_T( X_T - \widetilde x_{T-1} )
- C_T( x_T - \widetilde x_{T-1} )
&= C_T( \widetilde x_T - \widetilde x_{T-1} )
- C_T( \widetilde x_T \wedge x_T - \widetilde x_{T-1} )
\end{align*}
by the supermodularity of $F( \cdot, \theta_T )$ and the additive separability of $C_T$, and the final inequality holds since $\widetilde x_T$ maximizes $G_T( \cdot, \widetilde x_{T-1} )$ on $L$ by definition.

It remains to establish \eqref{eq:myopic-vs-fwd_induction_t}. Fix an arbitrary $t \geq T+1$. It suffices to show that $F( \widetilde x_T \vee x_t, \theta_t ) \geq F( x_t, \theta_t )$ and $C_t( \widetilde x_T \vee x_t - \widetilde x_T \vee x_{t-1} ) \leq C_t( x_t - x_{t-1} )$. The latter holds by \Cref{lemma:mon_impl2} (\cref{sec:appendix:mon_impl}) since $C_t$ is monotone. To show the former, begin by noting that $G_T(\widetilde x_T,\widetilde x_{T-1}) \geq G_T(\widetilde x_T \wedge x_t,\widetilde x_{T-1})$ since $\widetilde x_T$ maximizes $G_T(\cdot,\widetilde x_{T-1})$ on $L$ by definition. We have $C_T(\widetilde x_T - \widetilde x_{T-1}) \geq C_T(\widetilde x_T \wedge x_t - \widetilde x_{T-1})$ by \Cref{corollary:mon_impl} (\cref{sec:appendix:mon_impl}) since $\widetilde x_T \geq \widetilde x_{T-1}$. It follows that $F(\widetilde x_T,\theta_T) \geq F(\widetilde x_T \wedge x_t,\theta_T)$. Thus $F(\widetilde x_T \vee x_t,\theta_T) \geq F(x_t,\theta_T)$ by supermodularity, whence $F(\widetilde x_T \vee x_t,\theta_t) \geq F(x_t,\theta_t)$ by single-crossing differences and the fact that $\theta_t \geq \theta_T$ (since $t>T$ and $(\theta_s)_{s=1}^\infty$ is increasing). \qed

%%%%%%%%%%%%%%%%%%%%%%%%%%%%%%%%%%%
\subsection{Proof of \texorpdfstring{\hyperref[proposition:lechatelier_charac]{\Cref*{theorem:lechatelier}$\boldsymbol{^\dag}$}}{Theorem \ref{theorem:lechatelier}†}}
\label{sec:appendix:necessity:pf_lechatelier_charac}
%%%%%%%%%%%%%%%%%%%%%%%%%%%%%%%%%%%

The proof of \Cref{theorem:lechatelier} (\cref{sec:lechatelier}) shows that \ref{item:lechatelier_charac:weakmon} implies \ref{item:lechatelier_charac:mcs}. To show that \ref{item:lechatelier_charac:mcs} implies \ref{item:lechatelier_charac:weakmon}, we prove the contrapositive. If $C$ fails to be minimally monotone, then \ref{item:lechatelier_charac:mcs} fails by \hyperref[proposition:basic_charac]{\Cref*{theorem:basic}$^\dag$}. Suppose instead that there are adjustment vectors $\varepsilon,\varepsilon' \in \Delta L$ such that $C(\varepsilon') > C(\varepsilon)$ and either $0 \leq \varepsilon' \leq \varepsilon$ or $0 \geq \varepsilon' \geq \varepsilon$; we must show that \ref{item:lechatelier_charac:mcs} fails. Assume that $0 \leq \varepsilon' \leq \varepsilon$; the other case is analogous. Let $\alpha = C(\varepsilon') - C(\varepsilon) > 0$ and $\beta = \alpha + \max\{ 0, C(\varepsilon) - C(0) \}$. Note that $\varepsilon' < \varepsilon$ (else $\varepsilon' = \varepsilon$, which would imply $\alpha=0$).

Choose $\underline x, \bar x, \widehat x \in L$ such that $\widehat x - \underline x = \varepsilon$ and $\bar x - \underline x = \varepsilon'$, and note that $\underline x \leq \bar x < \widehat x$. Let $X = \{ \underline x, \bar x, \widehat x \}$; $X$ is a sublattice. Fix distinct $\underline \theta \leq \bar \theta$ in $\Theta$ (possible by hypothesis). Let $F(\underline x, \underline \theta) = 2$, $F(\bar x, \underline \theta) = 1$ and $F(\widehat x, \underline \theta) = 0$, and let $F(\underline x, \bar \theta) = 0$, $F(\bar x, \bar \theta) = \beta$ and $F(\widehat x, \bar \theta) = \beta - \alpha/2$.%
\footnote{For $\theta \in \Theta \setminus \{ \underline \theta, \bar \theta \}$, let $F(\cdot,\theta) = F(\cdot,\underline \theta)$ if $\theta \leq \underline \theta$ and $F(\cdot,\theta) = F(\cdot,\bar \theta)$ otherwise.}
Then $F(x,\theta)$ is supermodular in $x$ and has single-crossing differences in $(x,\theta)$, and we have
\begin{equation*}
\argmax_{x \in X} F(x,\underline \theta) = \{ \underline x \} , \quad
\argmax_{x \in X} G(x,\bar \theta) = \{ \widehat x \} ,\quad
\argmax_{x \in X} F(x,\bar \theta) = \{ \bar x \} .
\end{equation*}
Hence \ref{item:lechatelier_charac:mcs} fails, as $\bar \theta \geq \underline \theta$ and $\widehat x > \bar x$. \qed

%%%%%%%%%%%%%%%%%%%%%%%%%%%%%%%%%%%
\subsection{Proof of \texorpdfstring{\hyperref[proposition:lechatelier-dynamic_charac]{\Cref*{theorem:lechatelier-dynamic}$\boldsymbol{^\dag}$}}{Theorem \ref{theorem:lechatelier-dynamic}†}}
\label{sec:appendix:necessity:pf_lechatelier-dynamic_charac}
%%%%%%%%%%%%%%%%%%%%%%%%%%%%%%%%%%%

\ref{item:lechatelier-dynamic_charac:mon} implies \ref{item:lechatelier-dynamic_charac:mcs} by \Cref{theorem:lechatelier-dynamic}. To show that \ref{item:lechatelier-dynamic_charac:mcs} implies \ref{item:lechatelier-dynamic_charac:mon}, we prove the W: suppose that there is a $C \in \mathcal{C}$ that is not monotone, meaning that there are adjustment vectors $\varepsilon,\varepsilon' \in \Delta L$ such that $C(\varepsilon') > C(\varepsilon)$ and for each dimension $i$, either $0 \leq \varepsilon_i' \leq \varepsilon_i$ or $0 \geq \varepsilon_i' \geq \varepsilon_i$; we will show that \ref{item:lechatelier-dynamic_charac:mcs} fails.

Suppose first that $\varepsilon \geq 0$ or $\varepsilon \leq 0$. Then $C$ fails to be \emph{weakly} monotone, so by \hyperref[proposition:lechatelier_charac]{\Cref*{theorem:lechatelier}$^\dag$}, there exist a sublattice $X \subseteq L$, an objective $\widetilde F : X \times \Theta \to \R$ such that $\widetilde F(x,\theta)$ is quasi-supermodular in $x$ and has single-crossing differences in $(x,\theta)$, parameters $\underline \theta, \bar \theta \in \Theta$, and actions $\underline x \in \argmax_{x \in X} \widetilde F(x,\underline \theta)$ and $\bar x \in \argmax_{x \in X} \widetilde F(x,\bar \theta)$, such that $M = \argmax_{x \in X} \left[ \widetilde F(x,\bar \theta) - C(x-\underline x) \right]$ is non-empty and either (i)~$\underline \theta \leq \bar \theta$, $\underline x \leq \bar x$ and $M \cap \{ x \in X : \underline x \leq x \leq \bar x \} = \varnothing$ or (ii)~$\underline \theta \geq \bar \theta$, $\underline x \geq \bar x$ and $M \cap \{ x \in X : \underline x \geq x \geq \bar x \} = \varnothing$. Let $F = (1-\delta) \widetilde F$, $C_1 = C$, $C_t = \bar C$ for every $t \geq 2$, and $\theta_t = \bar \theta$ for every $t \in \N$. Then the long-lived agent's problem admits a solution, and every solution $(x_t)_{t=1}^\infty$ has $x_1 \in M$, so \ref{item:lechatelier-dynamic_charac:mcs} fails. Assume for the remainder that $0 \nleq \varepsilon \nleq 0$.

Note that $\varepsilon$ may satisfy $\varepsilon \wedge \varepsilon' = \varepsilon'$ or $\varepsilon \vee \varepsilon' = \varepsilon'$, but not both (else $\varepsilon' = \varepsilon$, which would imply $C(\varepsilon') = C(\varepsilon)$). Hence $E = \{ \varepsilon \wedge \varepsilon', \varepsilon, \varepsilon', \varepsilon \vee \varepsilon' \}$ has either two or four elements, and $( E \setminus \{\varepsilon'\} ) \cap \{ \widetilde \varepsilon \in \Delta L : \varepsilon' \wedge 0 \leq \widetilde \varepsilon \leq \varepsilon' \vee 0 \} = \varnothing$.%
    \footnote{If $\varepsilon \wedge \varepsilon' \neq \varepsilon'$, then $0 \geq \varepsilon_i' > \varepsilon_i$ for some dimension $i$, so $( \varepsilon \wedge \varepsilon' )_i = \varepsilon_i < \varepsilon_i' = ( \varepsilon' \wedge 0 )_i$, whence $\varepsilon \wedge \varepsilon' \ngeq \varepsilon' \wedge 0 \nleq \varepsilon$. Similarly, if $\varepsilon \vee \varepsilon' \neq \varepsilon'$ then $\varepsilon \vee \varepsilon' \nleq \varepsilon' \vee 0 \ngeq \varepsilon$.}

Choose $\widehat x \in L$ such that $\widehat x + \varepsilon \in L$. Let $X = \{ \widehat x \} + ( E \cup \{ 0, \varepsilon \wedge 0, \varepsilon \vee 0, \varepsilon' \wedge 0, \varepsilon' \vee 0 \} )$, where ``$+$'' denotes Minkowski addition;%
\footnote{That is, $A+B = \{ a+b : a \in A, b \in B \}$ for any sets $A,B \subseteq \R^n$.}
$X$ is a sublattice. Let $\underline x = \widehat x + \varepsilon' \wedge 0 \in X$ and $\bar x = \widehat x + \varepsilon' \vee 0 \in X$, and note that
\begin{equation}\label{emptyintersection}
\bigl( (\{\widehat x\} + E) \setminus \{ \widehat x + \varepsilon' \} \bigr) \cap \{ x \in X : \underline x \leq x \leq \bar x \} = \varnothing .
\end{equation}
Arrange the elements of $X$ in a matrix as
\begin{equation*}
\widehat x
+ \begin{pmatrix}
\varepsilon & \varepsilon \vee \varepsilon' & \varepsilon \vee 0 \\
\varepsilon \wedge \varepsilon' & \varepsilon' & \varepsilon' \vee 0 \\
\varepsilon \wedge 0 & \varepsilon' \wedge 0 & 0
\end{pmatrix}
=
\begin{pmatrix}
\widehat x + \varepsilon & \widehat x + \varepsilon \vee \varepsilon' & \widehat x + \varepsilon \vee 0 \\
\widehat x + \varepsilon \wedge \varepsilon' & \widehat x + \varepsilon' & \bar x \\
\widehat x + \varepsilon \wedge 0 & \underline x & \widehat x
\end{pmatrix} .
\end{equation*}
The top-left $2 \times 2$ submatrix is $\{\widehat x\} + E$. Here and below, ignore the leftmost column if $\varepsilon \wedge \varepsilon' = \varepsilon'$, since then its entries equal those of the middle column, and similarly ignore the top row if $\varepsilon \vee \varepsilon' = \varepsilon'$, the rightmost column if $\varepsilon' \wedge 0 = 0$, and the bottom row if $\varepsilon' \vee 0 = 0$.

Fix distinct $\underline \theta, \theta_1,\theta_2,\bar \theta \in \Theta$ such that $\underline \theta \leq \theta_1 \leq \bar \theta \geq \theta_2 \geq \underline \theta$ and $\theta_1 \nleq \theta_2 \nleq \theta_1$ (possible by hypothesis). Using the matrix notation from above, let
\begin{multline*}
F(\cdot,\underline \theta)
= \begin{pmatrix}
0 & 1 & 0 \\
1 & 2 & 1 \\
2 & 3 & 2
\end{pmatrix} , \quad
F(\cdot,\theta_1)
= \gamma \begin{pmatrix}
0 & 1 & 2 \\
1 & 2 & 3 \\
2 & 3 & 4
\end{pmatrix} , \quad
F(\cdot,\bar \theta)
= \begin{pmatrix}
0 & 1 & 2 \\
1 & 2 & 3 \\
0 & 1 & 2
\end{pmatrix} ,
\\
\text{and} \quad
F(\cdot,\theta_2)
= (1-\delta)
\begin{pmatrix}
0 & \alpha/3 & -\beta \\
\alpha/3 & 2\alpha/3 & \alpha/3-\beta \\
-\beta & \alpha/3-\beta & -2\beta
\end{pmatrix} ,
\end{multline*}
where $\alpha = C(\varepsilon') - C(\varepsilon) > 0$, $\beta = 2 \left[ \max_{\widetilde \varepsilon \in \Delta X} C(\widetilde \varepsilon) - \min_{\widetilde \varepsilon \in \Delta X} C(\widetilde \varepsilon) \right] \geq 2 \alpha$ and $\gamma = 3\beta$.%
\footnote{For $\theta \in \Theta \setminus \{ \underline \theta, \theta_1, \theta_2, \bar \theta \}$, let $F(\cdot,\theta) = F(\cdot,\underline \theta)$ if $\theta_1 \geq \theta \leq \theta_2$, $F(\cdot,\theta) = F(\cdot,\theta_1)$ if $\theta \leq \theta_1$ and $\theta_2 \nleq \theta \nleq \theta_2$, $F(\cdot,\theta) = F(\cdot,\theta_2)$ if $\theta \leq \theta_2$ and $\theta_1 \nleq \theta \nleq \theta_1$, and $F(\cdot,\theta) = F(\cdot,\bar \theta)$ if $\theta > \theta_1$ or $\theta > \theta_2$.}
Then $F(x,\theta)$ is supermodular in $x$ and has single-crossing differences in $(x,\theta)$, and we have
\begin{multline*}
\argmax_{x \in X} F(x,\underline \theta) = \{ \underline x \} , \quad
\argmax_{x \in X} F(x,\theta_1) = \{ \widehat x \} , \quad
\argmax_{x \in X} F(x,\bar \theta) = \{ \bar x \} ,
\\
\text{and} \quad \argmax_{x \in X} \left[ \left( \sum_{t=2}^\infty \delta^{t-1} \right) F(x,\theta_2) - \delta C(x-\widehat x) \right] \subseteq (\{\widehat x\} + E) \setminus \{ \widehat x + \varepsilon' \} .
\end{multline*}
The final inclusion holds because $\beta$ is large enough that $\widehat x+\varepsilon$ is strictly better than every $x \in X \setminus ( \{\widehat x\} + E )$ and because $\widehat x+\varepsilon$ is strictly better than $\widehat x + \varepsilon'$.

Let $C_1 = \underline C$, $C_2 = C$, and $C_t = \bar C$ and $\theta_t = \theta_2$ for every $t \geq 3$. Then the long-lived agent's problem admits a solution, and every solution $(x_t)_{t=1}^\infty$ has $x_1 = \widehat x$ and $x_2 \subseteq (\{\widehat x\} + E) \setminus \{ \widehat x + \varepsilon' \}$. Here $x_1 = \widehat x$ holds since $\gamma$ is large enough that it is optimal to choose myopically in period $t=1$. Hence \ref{item:lechatelier-dynamic_charac:mcs} fails, as $\underline \theta \leq \theta_t \leq \bar \theta$ for every $t \in \N$ and $x_2 \notin \{ x \in X : \underline x \leq x \leq \bar x \}$ by \eqref{emptyintersection}. \qed

%%% bibliography
\singlespacing
\printbibliography[heading=bibintoc]

\end{document}